\documentclass[journal,twoside,web]{ieeecolor}% Comment this line out if you need a4paper
\usepackage{generic}
\usepackage{algorithm}
\usepackage{algorithmicx}
\usepackage{algpseudocode}
\usepackage{graphicx}          % Include this line if your 
\usepackage{mathrsfs}         
\usepackage{amsmath}   
\usepackage{amsfonts}    
\usepackage{amssymb}   
\usepackage{subfigure}
\usepackage{comment}
\usepackage{threeparttable}
\usepackage{booktabs}
\usepackage{tabularx}
\usepackage{color}
\newtheorem{theorem}{Theorem}

\newtheorem{definition}{Definition}
\newtheorem{lemma}{Lemma}
\newtheorem{corollary}{Corollary}

\newtheorem{remark}{Remark}
\newtheorem{example}{Example}
\newtheorem{assumption}{Assumption}
\DeclareMathOperator{\rank}{rank}
\DeclareMathOperator{\mat}{mat}
\DeclareMathOperator{\abs}{abs}
\newcommand*{\QEDA}{\hfill\ensuremath{\blacksquare}}   
\newcounter{protocol}
\newenvironment{protocol}[1]
{\par\addvspace{\topsep}
	\noindent
	\tabularx{\linewidth}{@{} X @{}}
	\hline
	\refstepcounter{protocol}\textbf{Protocol \theprotocol} #1 \\
	\hline}
{ \\
	\hline
	\endtabularx
	\par\addvspace{\topsep}}
\def\BibTeX{{\rm B\kern-.05em{\sc i\kern-.025em b}\kern-.08em
		T\kern-.1667em\lower.7ex\hbox{E}\kern-.125emX}}
\begin{document}
\title{Angle-Based Sensor Network Localization}
\author{Gangshan Jing, Changhuang Wan and Ran Dai % <-this % stops a space
\thanks{This supplementary paper contains all the theoretical proofs missed in the journal paper ``Angle-Based Sensor Network Localization", which is published in IEEE Transactions on Automatic Control.} 
\thanks{Gangshan Jing and Changhuang Wan are with Department of Mechanical and Aerospace Engineering, The Ohio State University, Columbus, OH 43210, USA.
        Emails: {nameisjing@gmail.com and wan.326@osu.edu }}%
    \thanks{Ran Dai is with the School of Aeronautic and Astronautics, Purdue University, West Lafayette, IN 47907, USA. Email: {randai@purdue.edu}}
}

\maketitle
%\thispagestyle{empty}
%\pagestyle{empty}

%%%%%%%%%%%%%%%%%%%%%%%%%%%%%%%%%%%%%%%%%%%%%%%%%%%%%%%%%%%%%%%%%%%%%%%%%%%%%%%%
\begin{abstract}
This paper studies angle-based sensor network localization (ASNL) in a plane, which is to determine locations of all sensors in a sensor network, given locations of partial sensors (called anchors) and angle measurements obtained in the local coordinate frame of each sensor. Firstly it is shown that a framework with a non-degenerate bilateration ordering must be angle fixable, implying that it can be uniquely determined by angles between edges up to translations, rotations, reflections and uniform scaling. Then ASNL is proved to have a unique solution if and only if the grounded framework is angle fixable and anchors are not all collinear. Subsequently, ASNL is solved in centralized and distributed settings, respectively. The centralized ASNL is formulated as a rank-constrained semi-definite program (SDP) in either a noise-free or a noisy scenario, with a decomposition approach proposed to deal with large-scale ASNL. The distributed protocol for ASNL is designed based on inter-sensor communications. Graphical conditions for equivalence of the formulated rank-constrained SDP and a linear SDP,  decomposition of the SDP, as well as the effectiveness of the distributed protocol, are proposed, respectively. Finally, simulation examples demonstrate our theoretical results.

%This paper introduces a concept of ``angle fixability" and applies it to angle-based sensor network localization (ASNL). Generally speaking, a framework is said to be angle fixable if it can be uniquely determined by angles between edges up to translation, rotation, scaling, and reflection. It is shown that any framework with a ``non-degenerate bilateration ordering" is angle fixable in the plane. The ASNL problem is to determine locations of all sensors, given locations of partial sensors and angle constraints based on bearings measured in the local coordinate frame of each sensor. By establishing connections between angle fixability and angle localizability, it is proved that if a sensor network in the plane has a non-degenerate bilateration ordering and anchors are not all collinear, then the ASNL has a unique solution. If the network contains an acute-triangulated subframework, then the ASNL can be equivalently formulated as a linear semi-definite programming (SDP) to be solved within polynomial time. Simulation examples are provided to validate effectiveness of the theoretical findings.
\end{abstract}

\begin{IEEEkeywords}
	Network localization, angle rigidity, rank-constrained optimization, non-convex optimization, chordal decomposition
\end{IEEEkeywords}
\section{Introduction}

A sensor network localization problem is to determine locations of all sensors when locations of partial sensors (called anchors) and relative measurements between some pairs of sensors are available. It has received significant attention due to the importance of sensor locations in many scenarios, e.g., fusion of sensor measurements according to locations, searching sensors in specified areas, and tracking a moving target \cite{Aspnes06,Mao07,Paul17}. 

In the literature, depending on sensing capabilities of sensors, sensor network localization (SNL) has been studied via relative position-based \cite{Barooah08}, range-based \cite{Aspnes06}, \cite{Biswas04}-\cite{Wan19} and bearing (angle of arrival)-based \cite{Eren06}-\cite{Li19} approaches. Among them, bearing-based SNL (BSNL) is a popular topic in recent years since bearings can be captured by vision sensors \cite{Zhao18}. Nevertheless, BSNL requires each sensor to know bearing measurements with respect to the global coordinate frame, which can be realized by either equipping each sensor with specific devices (e.g., GPS, compasses) \cite{Zhu14,Zhao16} or implementing coordinate frame alignment algorithms \cite{Trinh18,Li19,Van18} via inter-sensor communications. As a result, these methods either become invalid in GPS-denied environments (e.g., underwater, indoor) or require frequent inter-sensor communications before or during implementation of the localization protocol. Although the authors in \cite{Lin16} proposed an algorithm based on bearings measured in local coordinate frames, the sensing graph has to contain more edges for solvability of SNL compared to localization via global bearing measurements (e.g., \cite{Zhu14,Zhao16}). In addition, extensive efforts have been carried out on range-based SNL (RSNL), where range measurements are independent of the global coordinate frame. Unfortunately, in many circumstances the solvability of RSNL requires more sensing than BSNL.

In SNL problems, it is important to distinguish what kind of sensor network is localizable given available anchor locations and measurements from sensors. This problem is usually tackled by checking whether the shape of the grounded graph can be uniquely determined by measurements. In BSNL and RSNL, bearing rigidity theory \cite{Eren06}-\cite{Li19}, \cite{Trinh18b,Trinh19} and distance rigidity theory \cite{Aspnes06,So07,Wan19}, \cite{Asimow78}-\cite{Henneberg11} are employed to propose conditions for localizability, respectively. %It has been shown that infinitesimal bearing rigidity and global distance rigidity are necessary and sufficient conditions for localizability of a SNL via bearing and range measurements, respectively. 
In \cite{Jing18}, the authors first developed an angle-based shape determination approach (namely, angle rigidity theory), where the minimum number of edges required for shape determination is the same as that for the bearing-based approach. Note that an angle between two edges joining one sensor is independent of the global coordinate frame. In practice, bearing (angle) measurements in a local coordinate system are usually low cost, reliable, and can be captured by vision sensors (e.g., monocular pinhole cameras \cite{Buckley17}). In recent years, angle-based formation control has attracted a growing interest due to the above-mentioned advantages of using angles as constraints or measurements \cite{Jing18}-\cite{Chen19}. However, the application of angle measurements to SNL has not been fully explored. Although a distributed SNL problem is equivalent to a distributed formation control problem in special cases, they are generally different because sensors may not be subject to dynamics constraints \footnote{In this paper, we study SNL from an optimization perspective, where sensors are not considered to have specific dynamics on their estimated states. In the case when sensors are subject to dynamics constraints, the distributed SNL problem becomes how to solve the optimizations formulated in this paper in a multi-agent cooperative control setting.}. In \cite{Fang20}, angle measurements are utilized in SNL, but the proposed approach requires the network to have more sensing and communication links than being angle rigid. In the present work, SNL based on angle  measurements will be studied under a milder graphical condition than that in \cite{Jing18}-\cite{Fang20}, and the proposed algorithms achieve guaranteed global convergence. SNL based on angle measurements is named angle-based sensor network localization (ASNL), which is to determine locations of the sensors other than anchors, given locations of anchors and angle measurements obtained in the local coordinate system of each sensor.

In this paper, we propose the concept of angle fixability based on angle rigidity theory in \cite{Jing18,Jing19} to characterize the property of a network that can be determined by angles uniquely up to translations, rotations, uniform scaling and reflections. By establishing connections between angle fixability and angle localizability, the results on angle fixability are applied to ASNL problems. ASNL will be studied in centralized and distributed frameworks, respectively. In both centralized ASNL (CASNL) and distributed ASNL (DASNL), each sensor is only capable of obtaining angle measurements with respect to its own local coordinate frame. %In the centralized framework, all sensors are able to transmit its measured information to a common central unit, while communications between different sensors are not required. For the central unit, the ASNL problem is modeled as a rank-constrained semi-definite programming problem that has been widely studied in the literature \cite{Sun17,Wan19,You19,Miller19}. Since a rank-constrained optimization problem is generally NP-hard, we further give a graph condition for removing the rank constraint by exploiting particular properties of ASNL. Although SDP with or without rank constraints have been solved in existing literature, it is time consuming to solve an optimization problem subject to a semi-definite constraint, especially for large scale SDPs. To overcome this challenge, we exploit the sparsity of the formulated SDP problem when the network consists of a large amount of sensor nodes. Under certain graph conditions, we further propose a decomposition approach for ASNL. CASNL in a noisy environment is also studied via a maximum likelihood formulation. In the distributed framework, there is no central unit required to receive information from all sensors. However, to cooperatively solve the ASNL problem, neighboring sensors should be able to communicate with each other.

A preliminary version of the centralized case has been presented in \cite{Jing19con}. This paper extends our former work in \cite{Jing19con} by presenting the generic property of angle fixability, chordal decomposition, ASNL with noises, DASNL and detailed proofs for Theorems \ref{th rank} and \ref{th D rank1}.

Our main contributions are summarized as follows:
\begin{itemize}
\item Equivalent algebraic conditions (Lemmas \ref{le f=ig}, \ref{le f=gn}) and a sufficient graphical condition  (Theorem \ref{th angle fixable}) for angle fixability in a plane are proposed. This graphical condition is milder than that in angle-based cooperative control references \cite{Jing18}-\cite{Fang20}, and it implicitly implies an approach to constructing angle fixable frameworks (Defnition \ref{de af framework}).
\item A graphical condition for localizability of ASNL is given (Theorem \ref{th localizability=fixability}). The CASNL problem is formulated as a rank-constrained semi-definite program (SDP) (Lemma \ref{le rank Z}). It is shown that if the grounded framework is acute-triangulated, then ASNL is equivalent to a linear SDP, which can be solved in polynomial time; see Theorem \ref{th equivalence}.
\item To handle large scale ASNL problems, we formulate ASNL as an SDP with two unknown matrices (problem (\ref{YD optimization})). When the grounded graph has a bilateration ordering, the first unknown matrix can be decomposed via chordal decomposition (Theorem \ref{th decomposeY}); when the grounded framework is acute-triangulated, the second unknown matrix can be decomposed into matrices in reduced sizes as well (Theorem \ref{th decomposeD}).
\item In a noisy environment, from the maximum likelihood estimation perspective, we model ASNL as an SDP with multiple rank-1 constraints and semi-definite constraints, which can be solved by algorithms in \cite{Sun17,You19,Miller19}.
\item Based on communications between adjacent sensors, a distributed protocol (Protocol 1) is proposed, which solves ASNL with guaranteed finite-time convergence (Theorem \ref{th BLP}). The upper bound of the convergence step is shown to be the number of sensors to be localized.
\end{itemize}
 
The main advantages of the proposed ASNL approach can be summarized from the following two perspectives: (i) Compared with BSNL, each sensor does not need bearing information in the global coordinate frame. In \cite{Trinh18,Li19,Van18}, sensors obtain  global bearing measurements by communicating with each other, which is a necessary procedure before or during implementation of the bearing-based localization algorithm. In constrast, the proposed CASNL approach does not require communications between sensors at all. Moreover, in DASNL, the coordinate frame alignment procedure can be avoided, and each sensor only communicates with its neighboring sensors for finite times. Hence, the ASNL approach requires lower communication costs. (ii) Compared with RSNL, the angle-based approach is applicable to a set of SNL that cannot be resolved by the existing range-based approaches (e.g., examples in Fig. \ref{fig angle rigidity} (b), (c), (d), and Fig. \ref{fig example for asnl}).

The outline of this paper is as follows. Section \ref{sec angle rigidity} provides preliminaries of angle rigidity theory and chordal decomposition. Section \ref{sec angle fixability} introduces the concept of angle fixability and provides criteria for angle fixability and relevant properties. Section \ref{sec ASNL} formulates ASNL as a QCQP and gives the necessary and sufficient conditions for ASNL to have a unique solution. Section \ref{sec centralized ASNL} solves the noise-free and noisy ASNL using a centralized framework. Section \ref{sec distributed ASNL} proposes a distributed protocol via inter-sensor communications for ASNL. Section \ref{sec simulation} exhibits several simulation examples. The concluding remarks are addressed in Section \ref{sec conclusion}.

\textbf{Notation}: Throughout the paper, $\mathcal{G}=(\mathcal{V},\mathcal{E})$ denotes an undirected graph, where $\mathcal{V}$ and $\mathcal{E}\subset\mathcal{V}\times\mathcal{V}$ denote the vertex set and edge set, respectively. The neighbor set of each vertex $i$ is denoted by $\mathcal{N}_i=\{j\in\mathcal{V}:(i,j)\in\mathcal{E}\}$. A $m\times n$ zero matrix is denoted by $\mathbf{0}_{m\times n}$, where ``$m\times n$" may be omitted if the dimension of the zero matrix can be observed. Given sets $A$ and $B$, $|A|$ is the cardinality of $A$, $A\setminus B$ is the set of elements in $A$ but not in $B$. The $d$-dimensional orthogonal group is written as $\text{O}(d)$. Given a matrix $X$, $\rank(X)$ is the rank of $X$, $X\succeq0$ implies that $X$ is positive semi-definite, $\det(X)$ denotes the determinant of $X$. A vector $p=(p_1^{\top},...,p_s^{\top})^{\top}$ is degenerate if $p_1,...,p_s$ are collinear. We use $\mathcal{K}$ to represent a complete graph with appropriate number of vertices, $I_d$ to denote the $d\times d$ identity matrix, $\otimes$ to denote the Kronecker product, $X_{a:b,c:d}$ is the submatrix of $X$ consisting of elements from $a$-th to $b$-th rows and $c$-th to $d$-th columns of $X$. Given matrices $X$ and $Y$, $\langle X,Y\rangle=\text{trace}(X^\top Y)$.

\section{Preliminaries}\label{sec angle rigidity}

In this section, some preliminaries of angle rigidity theory and chordal graphs will be introduced, which are important for studying  solvability and decomposability of an ASNL problem.

\subsection{Angle Rigidity Theory}

In \cite{Jing18}, angle rigidity theory is developed to study what kind of geometric shapes can be uniquely determined by angles subtended in the graph only. Similar to distance rigidity theory in RSNL and bearing rigidity theory in BSNL, angle rigidity theory plays an important role in solving ASNL. In \cite{Chen19}, the authors presented a different angle rigidity theory by taking the sign of each angle into account, which implies that all angles are defined in a common counterclockwise direction. Different from \cite{Chen19}, the angle considered in this paper does not have a specific sign. As a result, different sensors are allowed to have different definitions about the rotational direction. In this subsection, we will briefly review several definitions regarding angle rigidity theory proposed in \cite{Jing18} that will be used later. 

A graph $\mathcal{G}=(\mathcal{V},\mathcal{E})$ with $|\mathcal{V}|=n$ can be embedded in a plane by giving each vertex $i$ a position $p_i\in\mathbb{R}^2$. The vector $p=(p_1^{\top},...,p_n^{\top})^{\top}\in\mathbb{R}^{2n}$ is called a {\it configuration}, $(\mathcal{G},p)$ is called a {\it framework}. Each angle we use to determine the framework shape is an angle between two edges joining one common vertex, and the cosine of this angle will be constrained. For example, for the angle between edge $(i,j)$ and $(i,k)$, the cosine of this angle, i.e., $g_{ij}^{\top}g_{ik}$, will be constrained, where $g_{ij}=\frac{p_i-p_j}{||p_i-p_j||}$ is the bearing between vertices $i$ and $j$. The set of angle constraints in a graph $\mathcal{G}$ can be denoted by $\{g_{ij}^{\top}g_{ik}=a_{ijk}: a_{ijk}\in[-1,1],(i,j,k)\in\mathcal{T}_\mathcal{G}\}$, $\mathcal{T}_\mathcal{G}=\{(i,j,k)\in\mathcal{V}^3:(i,j),(i,k)\in\mathcal{E},j<k\}$, here ``$j<k$" avoids repeating each angle. Let $\theta_{ijk}$ denote the angle between $p_i-p_j$ and $p_i-p_k$, when $a_{ijk}$ is given, we can obtain a unique $\theta_{ijk}=\arccos a_{ijk}\in [0,\pi]$. That is, each angle constraint actually constrains an angle within the range $[0,\pi]$. Similar settings are considered in \cite{Eren03}-\cite{Jing19}. Note that when an angle is defined under a specified counterclockwise direction, it should be within the range $[0,2\pi)$ \cite{Chen19}. 

The {\it angle rigidity function} \cite{Jing18} of a framework $(\mathcal{G},p)$ is defined as
\begin{equation}
f_{\mathcal{G}}(p)=(..., g_{ij}^{\top}(p)g_{ik}(p),...)^{\top}, (i,j,k)\in\mathcal{T}_\mathcal{G}.
\end{equation}
A framework $(\mathcal{G},p)$ is {\it globally angle rigid} if $f_{\mathcal{G}}^{-1}(f_\mathcal{G}(p))=f_\mathcal{K}^{-1}(f_\mathcal{K}(p))$, here $\mathcal{K}$ is the complete graph with the same vertex set as $\mathcal{G}$. $(\mathcal{G},p)$ is {\it infinitesimally angle rigid} if all the infinitesimal angle motions are trivial. Here, the {\it infinitesimal angle motion} is a motion of the framework such that all angles in the framework (i.e., $f_\mathcal{G}(p)$) are invariant, a motion is {\it trivial} if it is a combination of  translations, rotations, and uniform scaling. An alternative condition for infinitesimally angle rigidity in $\mathbb{R}^2$ is $\rank(\frac{\partial f_{\mathcal{G}}(p)}{\partial p})=2n-4$. In \cite{Jing18}, the definitions of global angle rigidity and infinitesimal angle rigidity are based on existence of a subset of $\mathcal{T}_\mathcal{G}$, which are actually equivalent to our definitions here. Compared with bearing rigidity \cite{Zhao16}, distance rigidity \cite{Asimow78}, and weak rigidity \cite{Kwon18,Jing18w}, the essential novelty of angle rigidity theory is that only subtended angles are used in the rigidity function.

Three examples are presented in Fig. \ref{fig angle rigidity} to illustrate these definitions. In Fig. \ref{fig angle rigidity}, frameworks (a) and (e) are nonrigid. In framework (a), vertices 1, 2, 3 and 4 can move simultaneously to deform the shape while maintaining all subtended angles. In framework (e), vertices 4 and 5 can move freely along the line between 1 and 4 and the line between 2 and 5, respectively. Frameworks (b), (c) and (d) are globally and infinitesimally angle rigid because the angles in each framework are sufficient to determine the entire shape uniquely; in Fig. \ref{fig angle rigidity} (f), since the graph is complete, the framework is globally angle rigid. It is not infinitesimally angle rigid because vertex 2 can move freely along the line between vertices 1 and 3.

\begin{figure}
	\centering
	\includegraphics[width=9cm]{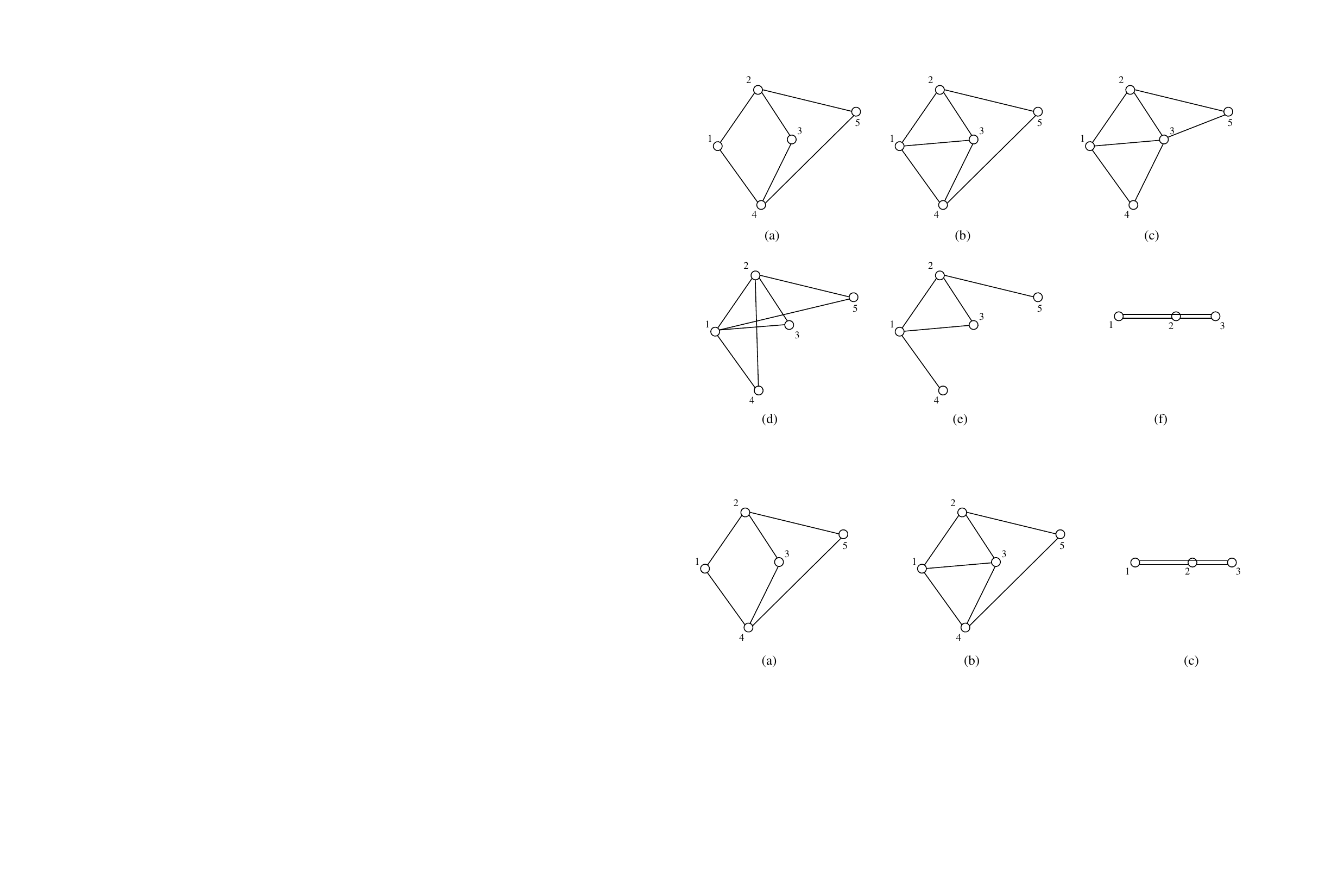}
	\caption{Some frameworks (graphs) in the plane, frameworks (a)-(e) have the same configuration but different graphs.}
	\label{fig angle rigidity}
\end{figure}

\subsection{Chordal Graphs and Chordal Decomposition}

A graph is said to be {\it chordal} if each cycle with more than three vertices in this graph has a {\it chord}. Here a chord is an edge between two nonconsecutive vertices in the cycle. A {\it clique} $\mathcal{C}$ of a graph $\mathcal{G}=(\mathcal{V},\mathcal{E})$ is a subset of $\mathcal{V}$ such that each pair of vertices in $\mathcal{C}$ are adjacent. In Fig. \ref{fig angle rigidity}, graphs (a) and (b) are not chordal, graphs (c)-(f) are all chordal.  We say a clique $\mathcal{C}$ is a {\it $s$-point clique} if $|\mathcal{C}|=s$. A clique $\mathcal{C}$ is said to be a {\it maximal clique} if there is no other clique containing this clique. In Fig. \ref{fig angle rigidity} (c), there are 3 maximal cliques: $\mathcal{C}_1=\{1,2,3\}$, $\mathcal{C}_2=\{1,3,4\}$, $\mathcal{C}_3=\{2,3,5\}$. Given a maximal clique $\mathcal{C}$, we define a transformation matrix $Q_{\mathcal{C}}\in\mathbb{R}^{|\mathcal{C}|\times n}$ such that $Q_{\mathcal{C}}\eta=(\eta_{\mathcal{C}(1)},...,\eta_{\mathcal{C}(|\mathcal{C}|)})^{\top}\in\mathbb{R}^{|\mathcal{C}|}$ for any $n$-dimensional vector $\eta=(\eta_1,...,\eta_n)^{\top}\in\mathbb{R}^n$, $\mathcal{C}(i)$ denotes the $i$-th element of $\mathcal{C}$. Each element of $Q_\mathcal{C}$ is defined as:

\begin{equation}  
(Q_\mathcal{C})_{ij}=
\begin{cases}
1, &  j=\mathcal{C}(i),\\  
0, &  \text{otherwise}.    
\end{cases}  
\end{equation} 
The following lemma gives a condition for equivalence between positive semi-definiteness of a matrix and positive semi-definiteness of its submatrices corresponding to maximal cliques.
\begin{lemma}\cite{Grone84}\label{le chordal decomposition}
	Given $\mathcal{G}=(\mathcal{V},\mathcal{E})$ as a chordal graph and a matrix $X\in\mathbb{R}^{|\mathcal{V}|\times|\mathcal{V}|}$, let $\{\mathcal{C}_1, \mathcal{C}_2,..., \mathcal{C}_p\}$ be the set of its maximal clique sets. Then, $X\succeq 0$ if and only if $Q_{\mathcal{C}_k}XQ_{\mathcal{C}_k}^{\top}\succeq0$, $k=1,...,p$.
\end{lemma}

\section{Angle fixability}\label{sec angle fixability}
To better understand what kind of geometric shapes can be uniquely determined by angles, we introduce the notion of {\it angle fixability} in this section, which is stronger than global angle rigidity and infinitesimal angle rigidity.

The formal definition of angle fixability is given below.

\begin{definition}\label{de fixability}
	A framework $(\mathcal{G},p)$ is angle fixable in $\mathbb{R}^d$ if $f_\mathcal{G}^{-1}(f_\mathcal{G}(p))=\mathscr{S}_p$, where
\begin{equation}\label{S_p}
\begin{split}
\mathscr{S}_p=\{q\in\mathbb{R}^{nd}: &q=c(I_n\otimes\mathscr{R})p+\mathbf{1}_n\otimes\xi, \\ &\mathscr{R}\in\text{O}(d),c\in\mathbb{R}\setminus\{0\}, \xi\in\mathbb{R}^d\}.
\end{split}
\end{equation}	
\end{definition}

From Definition \ref{de fixability}, we observe that the set $\mathscr{S}_p$ actually defines a set of configurations forming the same shape as the one formed by $p$. That is, if $q\in\mathscr{S}_p$, then $q$ can be obtained from $p$ by a combination of rotations, translations, uniform scaling and reflections. In this paper, we mainly focus on angle fixability in $\mathbb{R}^2$. Since the definition of angle fixability in $\mathbb{R}^d$ with $d\geq 3$ will be used in Lemma \ref{le universal fixability} and Theorem \ref{th rank} that state important conditions for removing the rank constraint in ASNL (Theorem \ref{th equivalence}), angle fixability is defined in an arbitrary dimensional space in Definition \ref{de fixability}.

\subsection{Equivalent Conditions for Angle Fixability in $\mathbb{R}^2$}

The following two lemmas give two necessary and sufficient conditions for angle fixability in $\mathbb{R}^2$. 

\begin{lemma}\label{le f=ig}
	In $\mathbb{R}^2$, $(\mathcal{G},p)$ is angle fixable if and only if it is globally and infinitesimally angle rigid.
\end{lemma}

\begin{proof}
	The sufficiency has been proven in \cite[Theorem 1]{Jing18}, next we prove the necessity. Suppose that $(\mathcal{G},p)$ is angle fixable but not infinitesimally angle rigid, then $\rank(\frac{\partial f_\mathcal{G}}{\partial p})=s< 2n-4$. As a result, there exists a neighborhood of $p$ in which $f_{\mathcal{G}}^{-1}(f_{\mathcal{G}}(p))$ is a $2n-s>4$ dimensional manifold, this conflicts with the fact that $f_{\mathcal{G}}^{-1}(f_{\mathcal{G}}(p))=\mathscr{S}_p$ is a 4-dimensional manifold. For global angle rigidity, since it always holds that $f_{\mathcal{K}}^{-1}(f_{\mathcal{K}}(p))\subset f_{\mathcal{G}}^{-1}(f_{\mathcal{G}}(p))$, it suffices to prove $f_{\mathcal{G}}^{-1}(f_{\mathcal{G}}(p))\subset f_{\mathcal{K}}^{-1}(f_{\mathcal{K}}(p))$. For any $q\in f_{\mathcal{G}}^{-1}(f_{\mathcal{G}}(p))$, we have $q\in\mathscr{S}_p$, then $g_{ij}^{\top}(q)g_{ik}(q)=g_{ij}^{\top}(p)g_{ik}(p)$ for all $i,j,k\in\mathcal{V}$. That is, $(\mathcal{G},p)$ is globally angle rigid. 
\end{proof}

\begin{lemma}\label{le f=gn}
	In $\mathbb{R}^2$, $(\mathcal{G},p)$ is angle fixable if and only if it is globally angle rigid and $p$ is non-degenerate.
\end{lemma}
\begin{proof}
	Since the configuration of an infinitesimally angle rigid framework can never be degenerate, the necessity can be obtained by Lemma \ref{le f=ig}. Next we prove sufficiency. Global angle rigidity implies that $f_{\mathcal{G}}^{-1}(f_{\mathcal{G}}(p))= f_{\mathcal{K}}^{-1}(f_{\mathcal{K}}(p))$, hence we only have to show that there exists some subgraph $\mathcal{G}'$ of $\mathcal{K}$ such that $(\mathcal{G}',p)$ is infinitesimally angle rigid. Without loss of generality, let $1$, $2$, $3$ be three vertices not lying collinear. We start with the complete graph with vertices 1, 2, 3, which is angle fixable. Note that for any $4\leq i\leq n$, there always exist two vertices $j,k\in\{1,2,3\}$ such that $p_i-p_j$ and $p_i-p_k$ are not collinear. By adding vertex $i$ and edges $(i,j)$, $(i,k)$ for $i=4,...,n$ iteratively, we obtain a new graph $\mathcal{G}'$. Moreover, at each step during the generation, the conditions in Lemma \ref{le fixability in R2} (will be proposed later with its proof independent of this lemma) are satisfied. Thus $(\mathcal{G}',p)$ is angle fixable. 
\end{proof}

Combining Lemma \ref{le f=ig} and Lemma \ref{le f=gn}, the following lemma holds.

\begin{lemma}
	Consider a globally angle rigid framework $(\mathcal{G},p)$ in $\mathbb{R}^2$, the following statements are equivalent:\\
	(i) $p$ is non-degenerate;\\
	(ii) $(\mathcal{G},p)$ is infinitesimally angle rigid;\\
	(iii) $(\mathcal{G},p)$ is angle fixable.
\end{lemma}

\subsection{Generic Angle Fixability}
In \cite{Jing18}, the authors showed that both infinitesimal angle rigidity and global angle rigidity are generic properties of the graph. That is, given a graph $\mathcal{G}$, either for all generic configurations\footnote{A configuration $p=(p_1^{\top},\cdots,p_n^{\top})^{\top}\in\mathbb{R}^{2n}$ is generic if its $2n$ coordinates are algebraically independent \cite{Jing18}.} $p\in\mathbb{R}^{2n}$, $(\mathcal{G},p)$ is infinitesimally (globally) angle rigid, or none of them is. Therefore, angle fixability in $\mathbb{R}^2$ is a generic property of the graph due to Lemma \ref{le f=ig}. We give the following definition and result.
\begin{definition}
	A graph $\mathcal{G}$ is generically angle fixable in $\mathbb{R}^2$ if $(\mathcal{G},p)$ is angle fixable for any generic configuration $p\in\mathbb{R}^{2n}$.
\end{definition}

\begin{lemma}\label{le generic}
	If $(\mathcal{G},p)$ is angle fixable for a generic configuration $p\in\mathbb{R}^{2n}$, then $\mathcal{G}$ is generically angle fixable in $\mathbb{R}^2$. 
\end{lemma}

We also note that all generic configurations in $\mathbb{R}^2$ form a dense space. Therefore, for a generically angle fixable graph $\mathcal{G}$, the set of all configurations $p\in\mathbb{R}^{2n}$ such that $(\mathcal{G},p)$ is not angle fixable is of measure zero. Lemma \ref{le generic}, together with Lemma \ref{le f=gn}, imply that for a framework with a generic configuration $p\in\mathbb{R}^{2n}$, angle fixability and global angle rigidity are equivalent. We summarize this result in the following lemma.

\begin{lemma}
	A graph $\mathcal{G}$ is generically angle fixable in $\mathbb{R}^2$ if and only if it is generically globally angle rigid in $\mathbb{R}^2$.
\end{lemma}

\subsection{Recognizing Angle Fixable Frameworks}

In this subsection, a graphical approach to recognizing angle fixable frameworks will be presented.  Before showing that, we firstly present the following result for angle fixable frameworks.
\begin{lemma}\label{le fixability in R2}
	Given an angle fixable framework in $\mathbb{R}^2$, after adding a node and two non-collinear edges connecting this node to two existing nodes, the induced framework is still angle fixable in $\mathbb{R}^2$.
\end{lemma}

\begin{proof}
	Let $(\mathcal{G},p)$ be the angle fixable framework with $n$ vertices, $n+1$ be the added node, $(n+1,u)$ and $(n+1,v)$ be the two added edges, $(\mathcal{G}',p')$ be the induced framework. We only need to verify $f_{\mathcal{G}'}^{-1}(f_{\mathcal{G}'}(p'))=\mathscr{S}_{p'}$ in order to prove that $(\mathcal{G}',p')$ is still angle fixable. Since it always holds that $\mathscr{S}_{p'}\subset f_{\mathcal{G}'}^{-1}(f_{\mathcal{G}'}(p'))$, it suffices to show $f_{\mathcal{G}'}^{-1}(f_{\mathcal{G}'}(p'))\subset\mathscr{S}_{p'}$. 
	
	For each $q'\in f_{\mathcal{G}'}^{-1}(f_{\mathcal{G}'}(p'))$, it must hold that $q'=(q^{\top},{q'_{n+1}}^{\top})^{\top}\in\mathbb{R}^{2n+2}$, where $q\in\mathscr{S}_{p}$, $q'_{n+1}$ satisfies $g_{i,n+1}^{\top}(q')g_{ij}(q')= g_{i,n+1}^{\top}(p')g_{ij}(p')$, $i\in\{u, v\}$, $j\in\mathcal{N}_i$.
	Next we show that $g_{u,n+1}$ can be uniquely determined by $q$ and $f_{\mathcal{G}'}(p')$. Lemma \ref{le f=ig} shows that $(\mathcal{G},p)$ is infinitesimally angle rigid. Then vertex $u$ must have at least two neighbors $j_1$, $j_2$ such that $p_u-p_{j_1}$ and $p_u-p_{j_2}$ are not collinear. Denote $A=(g_{uj_1},g_{uj_2})\in\mathbb{R}^{2\times 2}$, then $\rank(A)=2$. Note that if we regard $g_{u,n+1}=x=(x_1,x_2)^{\top}\in\mathbb{R}^2$ as unknown variables, we then have $A^{\top}x=b$, where $b=(g^{\top}_{u,n+1}(p')g_{uj_1}(p'), g^{\top}_{u,n+1}(p')g_{uj_2}(p'))$. Hence $g_{u,n+1}(p')$ can be uniquely determined by $q$ and $f_{\mathcal{G}'}(p')$. 
	
	Similarly, $g_{v,n+1}$ can be uniquely determined by $q$ and $f_{\mathcal{G}'}(p')$. Since $g_{u,n+1}$ and $g_{v,n+1}$ are not collinear, they have only one intersection point. As a result, $q'_{n+1}$ can be uniquely determined. Note that there must exist $\tilde{q}=(q^{\top},q_{n+1}^{\top})^{\top}\in\mathscr{S}_{p'}$ such that $\tilde{q}\in f_{\mathcal{G}'}^{-1}(f_{\mathcal{G}'}(p'))$, we then have $q'=\tilde{q}\in\mathscr{S}_{p}$.
\end{proof}

In two-dimensional (2D) space, it is well known that any minimally rigid framework is embedded by a Laman graph \cite{Laman70}, which can be obtained by Henneberg constructions \cite{Trinh18b,Trinh19,Henneberg11}. At each step of Henneberg construction, either one vertex and two new edges are added (named vertex addition), or one vertex and three new edges are added, while an existing edge is removed (named edge splitting). By Lemma \ref{le fixability in R2}, the specified Henneberg vertex additions preserve angle fixability in 2D space. In \cite{Fang09}, a graph containing a subgraph induced by Henneberg vertex addition (\cite{Trinh18b,Trinh19,Henneberg11}) is said to have a {\it bilateration ordering}.

For frameworks generated by graphs with a bilateration ordering, we define the non-degenerate bilateration ordering as follows.
\begin{definition}\label{de af framework}(Non-degenerate Bilateration Ordering)
A framework $(\mathcal{G}_b(n),p(n))$ is said to have a {\it non-degenerate bilateration ordering} if it can be generated by the following procedure: Starting with the 3-vertex framework $(\mathcal{G}_b(3),p(3))$ where $p(3)$ is non-degenerate, $(\mathcal{G}_b(i+1),p(i+1))$ is obtained from $(\mathcal{G}_b(i),p(i))$ by adding one vertex $l$ and $s\geq2$ edges connecting $l$ to existing vertices $l_1,..., l_s$ such that $p_l-p_{l_j}$, $j\in\{1,..,s\}$ are not all collinear.
\end{definition}

\begin{figure}
	\centering
	\includegraphics[width=9cm]{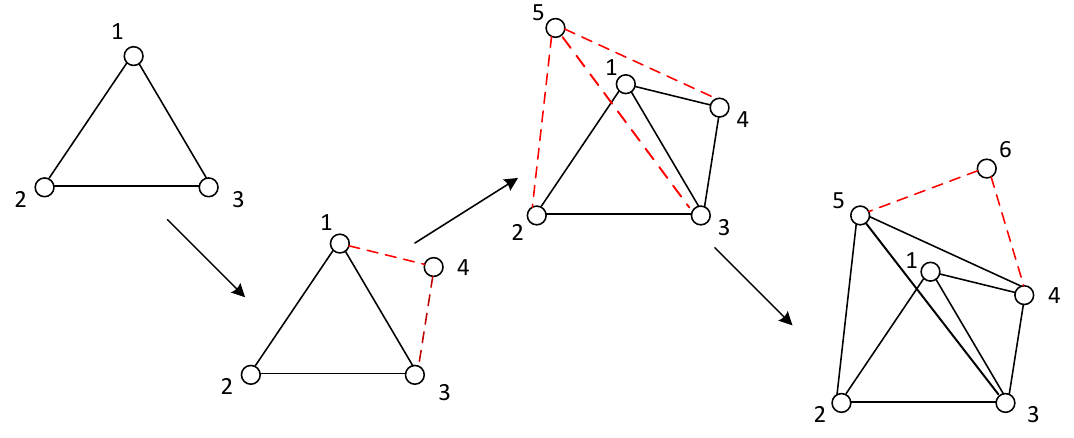}
	\caption{An example of non-degenerate bilateration ordering.}
	\label{fig bilateration ordering}
\end{figure}

Fig. \ref{fig bilateration ordering} shows an example of the non-degenerate bilateration ordering. From Lemma \ref{le fixability in R2}, we have the following result.

\begin{theorem}\label{th angle fixable}
In $\mathbb{R}^2$, if a framework has a non-degenerate bilateration ordering, then it is angle fixable.
\end{theorem}

A strongly non-degenerate triangulated framework $(\mathcal{G}_t,p)$ is a framework with a non-degenerate bilateration ordering where at each step when a vertex $l$ is added to $(\mathcal{G}_t(i),p(i))$, two non-collinear edges connecting $l$ to $j$ and $k$ such that $(j,k)\in\mathcal{E}_t$ are added accordingly. In \cite{Jing18}, $(\mathcal{G}_t,p)$ is shown to be angle fixable in $\mathbb{R}^2$. One can realize that a strongly non-degenerate triangulated framework always has a non-degenerate bilateration ordering, but not vice versa. Fig. \ref{fig angle rigidity} (b) shows a framework with a non-degenerate bilateration ordering, while it is not a triangulated framework, because vertices 3 and 5 are not adjacent. 

By generic property of angle fixability, the following result holds.

\begin{corollary}
	A graph with a bilateration ordering is generically angle fixable in $\mathbb{R}^2$.
\end{corollary}

An interesting fact is that even if the framework generated by Definition \ref{de af framework} in $\mathbb{R}^2$ is elevated into a higher dimensional space, it is still angle fixable. See the following lemma. 

\begin{lemma}\label{le universal fixability}
Given a framework $(\mathcal{G}_b(n),p)$ with a non-degenerate bilateration ordering in $\mathbb{R}^2$, for any integer $d\geq3$, $(\mathcal{G}_b(n),\bar{p})$ is angle fixable in $\mathbb{R}^d$, where $\bar{p}=(\bar{p}_1^{\top},...,\bar{p}_n^{\top})^{\top}$, $\bar{p}_i=(p_i^{\top},\mathbf{0}_{1\times(d-2)})^{\top}\in\mathbb{R}^{d}$.
\end{lemma}

\begin{proof} 
We prove the result by showing that in $\mathbb{R}^d$, every $q\in f_{\mathcal{G}_b(n)}^{-1}f_{\mathcal{G}_b(n)}(\bar{p})$ satisfies $q\in\mathscr{S}_{\bar{p}}$. Note that $1$, $2$ and $3$ always form a non-degenerate triangle, and a triangle is always angle fixable in any dimensional space (The shape of a non-degenerate triangle constrained by angles is invariant to the dimension of the space). Hence, there must hold that $q_i=c\mathscr{R}\bar{p}_i+\xi$ for some appropriate $c$, $\mathscr{R}$ and $\xi$, $i=1,2,3$. 
	
Suppose originally we have $(\mathcal{G}_b(i),q(i))$, where $q(i)\in\mathscr{S}_{p(i)}$ for $c$, $\mathscr{R}$ and $\xi$. Next, we prove that in the following Henneberg vertex addition introduced in Definition \ref{de af framework}, the position of the new vertex to be added can be uniquely determined by angle constraints and positions of existing vertices. Let $l$ be the added vertex, $(l,j)$ and $(l,k)$ be the two non-collinear added edges, $j_1$ and $j_2$ are two neighbors of $j$ in $\mathcal{G}_b(i)$ and $q_j-q_{j_1}$ is not collinear with $q_j-q_{j_2}$. From the angle constraints involving $j$, we have $g_{jl}^{\top}(q)g_{jj_1}(q)=c_{jlj_1}$ and $g_{jl}^{\top}(q)g_{jj_2}(q)=c_{jlj_2}$. Next we show $g_{jl}$ can be uniquely determined. Note that $g_{jl}^{\top}(q)g_{jj_1}(q)=g_{jl}^{\top}(q)\mathscr{R}g_{jj_1}(p)$, $g_{jl}^{\top}(q)g_{jj_2}(q)=g_{jl}^{\top}(q)\mathscr{R}g_{jj_2}(p)$, $\bar{p}_s=(p_s^{\top},\mathbf{0}_{1\times (d-2)})^{\top}$ for $s=j, j_1, j_2$. Denote $\mathscr{R}^{\top}g_{jl}(q)$ by $b=(b_1,..,b_d)^{\top}\in\mathbb{R}^d$, $g_{jj_1}(p)$ by $\zeta=(\zeta_1,\zeta_2,\mathbf{0}_{1\times (d-2)})^{\top}$, and $g_{jj_2}(p)$ by $\eta=(\eta_1,\eta_2,\mathbf{0}_{1\times(d-2)})^{\top}$. Then we have $\zeta^{\top}b=c_{jlj_1}$, $\eta^{\top}b=c_{jlj_2}$, which is equivalent to $\begin{pmatrix}
	\zeta_1 & \zeta_2\\
	\eta_1  & \eta_2
	\end{pmatrix}\begin{pmatrix}
	b_1 \\ b_2
	\end{pmatrix}$ =$\begin{pmatrix}
	c_{jlj_1} \\ c_{jlj_2}
	\end{pmatrix}$. Since $g_{jj_1}(p)$ and $g_{jj_2}(p)$ are not collinear, $b_1$ and $b_2$ can be uniquely determined. Note that $(b_1,b_2,\mathbf{0}_{1\times(d-2)})^{\top}=g_{lj}(p)$, implying that $b_1^2+b_2^2=1$. Since $||b||=1$, we have $b_s=0$, $s=3,...,d$. Then $g_{jl}(q)=\mathscr{R}b$ is uniquely determined. Similarly, $g_{lk}(q)$ is also unique. Recall that $g_{jl}$ and $g_{lk}$ are not collinear, they have a unique intersection point, i.e., $q_l$. This completes the proof. 
\end{proof}

The above lemma implies that the angle fixability of a framework with a non-degenerate bilateration ordering in $\mathbb{R}^2$ is invariant to space dimensions. However, it does not mean that any framework generated by Definition \ref{de af framework} in $\mathbb{R}^d$ is angle fixable.

\section{Angle-Based Sensor Network Localization}\label{sec ASNL}

In this section, the problem settings and the mathematical formulation for ASNL  will be presented. The condition in terms of the sensing graph for an ASNL problem to have a unique solution will be proposed as well.
\subsection{Problem Formulation}

To introduce the problem formulation of ASNL, we will explain how the sensor network is modelled; what kind of information each sensor senses; and how a general ASNL can be mathematically formulated, successively.

\subsubsection{Sensor network modelling}
Given a network of sensors indexed by $\mathcal{V}=\{1,...,n\}=\mathcal{A}\cup\mathcal{S}$, where $\mathcal{A}=\{1,...,n_a\}$, $\mathcal{S}=\{n_a+1,...,n_a+n_s\}$. Sensors in $\mathcal{A}$ are called {\it anchors}, whose locations are available. Sensors in $\mathcal{S}$ are called {\it unknown sensors}, whose locations are unknown and to be determined. An undirected graph $\mathcal{G}=(\mathcal{V},\mathcal{E})$ is used as the sensing graph interpreting the interaction relationships between sensors. Each sensor has the capability of sensing bearing measurements in its local coordinate frame from other neighboring sensors $j\in\mathcal{N}_i$. An ASNL problem in $\mathbb{R}^d$ is to determine $x_i$, $i\in\mathcal{S}$ when $\{x_i\in\mathbb{R}^d: i\in\mathcal{A}\}$ and all the angles between edges in $\mathcal{G}$ are available. In this paper, we will focus on the case with $d=2$.

Let $x=(x_1^{\top},...,x_n^{\top})^{\top}$. Note that for any two anchors $i, j$, even if $(i,j)\notin\mathcal{E}$, we can still obtain $(x_i-x_j)/||x_i-x_j||$ since we know the accurate values of $x_i$ and $x_j$. Therefore, it is reasonable to utilize all the angles in framework $(\hat{\mathcal{G}},x)$, with $\hat{\mathcal{G}}=(\mathcal{V},\hat{\mathcal{E}})$, $\hat{\mathcal{E}}=\mathcal{E}\cup\{(i,j)\in\mathcal{V}^2: i,j\in\mathcal{A}\}$ in solving SNL. In this paper, we call $(\hat{\mathcal{G}},x)$ the {\it grounded framework}, and use $\mathbf{N}=(\hat{\mathcal{G}},x,\mathcal{A})$ to denote a sensor network. 

A simple example of ASNL is shown in Fig. \ref{fig example for asnl}, where the position of sensor $4$ can be uniquely determined by $\cos\angle A$, $\cos\angle B$ and $\cos\angle C$. More specifically, $\cos\angle B$ and $\cos\angle C$ determine the shape of the triangle formed by 1, 3 and 4, the distance between anchors 1 and 3 determines the size of this triangle, $\cos\angle A$ determines the direction of the relative position between 1 and 4. Similarly, the position of sensor $5$ can be uniquely determined by angles $\cos\angle D$, $\cos\angle E$ and $\cos\angle F$. Note that both the locations of sensor 4 and sensor 5 cannot be determined by lengths of edges in Fig. \ref{fig example for asnl}, implying that the RSNL approach is not applicable to this example.

\begin{figure}
	\centering
	\includegraphics[width=6cm]{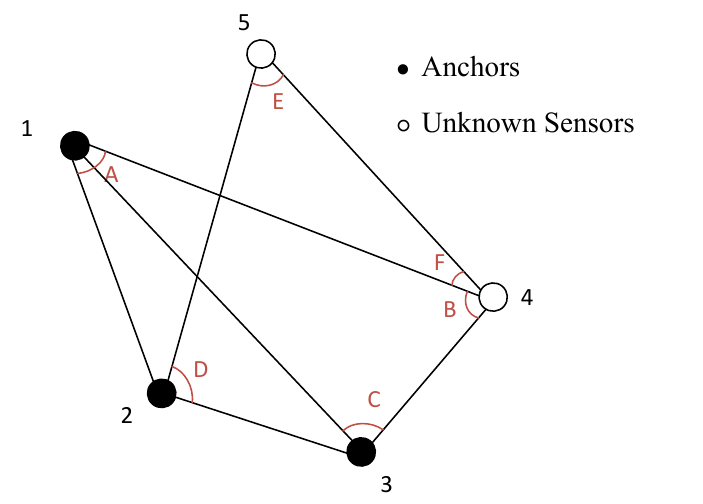}
	\caption{A sensor network where the locations of unknown sensors can be determined by measured angles and given locations of anchors.}
	\label{fig example for asnl}
\end{figure}

\subsubsection{Sensing capability}

In this paper, we consider that each sensor only senses relative bearing measurements from its immediate neighbors, which can be captured by vision-based sensors \cite{Buckley17,Zhao18}. Moreover, we consider a GPS-denied environment, in which each sensor has an independent coordinate system. As a result, each pair of sensors may have different understandings about their relative bearing. Such a setting is actually equivalent to assuming that each sensor senses angles subtended at itself.

%In fact, it is interesting to observe that for each sensor $i$, the angle between any pair of edges joining $i$ is invariant to different coordinate frames. In addition, such angle can always be computed by bearings measured in the local coordinate frame. 
\begin{comment}To explain this phenomenon mathematically, let $x^i_j$ denote the coordinate of sensor $j$ expressed in the local coordinate frame of $i$, $x_i$ be the coordinate of $i$ in the global coordinate frame. Then $x^i_j=\mathscr{R}_ix_j+\xi_i$ for some $\mathscr{R}_i\in \text{O}(2)$, $\xi_i\in\mathbb{R}^2$. It follows that
\[
\begin{split}
\frac{(x^i_i-x^i_j)^{\top}}{||x^i_i-x^i_j||}\frac{(x^i_i-x^i_k)}{||x^i_i-x^i_k||}&=\frac{(x_i-x_j)^{\top}}{||x_i-x_j||}\mathscr{R}_i^{\top}\mathscr{R}_i\frac{(x_i-x_k)}{||x_i-x_k||}\\
&=\frac{(x_i-x_j)^{\top}}{||x_i-x_j||}\frac{(x_i-x_k)}{||x_i-x_k||}.
\end{split}
\]
\end{comment}

\subsubsection{A QCQP formulation}

Given a sensor network $(\hat{\mathcal{G}},x,\mathcal{A})$ in $\mathbb{R}^2$, let $p=(p_1^{\top},...,p_n^{\top})^{\top}\in\mathbb{R}^{2n}$ be the real locations of sensors. The ASNL problem can be modeled as the following quadratically constrained quadratic program (QCQP):
\begin{equation}\label{qcqp}
\begin{split}
\text{find}~~ x,d_{ij}&,(i,j)\in\hat{\mathcal{E}}\\
\text{s.t.} ~~(x_i-x_j)^{\top}(x_i-x_k)&=a_{ijk}d_{ij}d_{ik},  (i,j,k)\in\mathcal{T}_{\hat{\mathcal{G}}}\\
||x_i-x_j||^2&=d_{ij}^2, ~~~~~~~~~~(i,j)\in\hat{\mathcal{E}}\\
x_i&=p_i, ~~~~~~~~~~~~~~~i\in\mathcal{A}
\end{split}
\end{equation}
where $a_{ijk}=\frac{(p_i-p_j)^{\top}}{||p_i-p_j||}\frac{(p_i-p_k)}{||p_i-p_k||}$ is the angle information obtained from bearing measurements, $\mathcal{T}_{\hat{\mathcal{G}}}=\{(i,j,k)\in\mathcal{V}^3:$ $(i,j),(i,k)\in\hat{\mathcal{E}}, j<k\}$ is the angle index set determining all the angles subtended in the framework, $d_{ij}$ is the distance between sensors $i$ and $j$ for $(i,j)\in\hat{\mathcal{E}}$. 
The known quantities in (\ref{qcqp}) include: $a_{ijk}\in\mathbb{R}$ for $(i,j,k)\in\mathcal{T}_{\hat{\mathcal{G}}}$, $p_i\in\mathbb{R}^2$ for $i\in\mathcal{A}$; the unknown variables in (\ref{qcqp}) are: $x_i\in\mathbb{R}^2$ for $i\in\mathcal{S}$, $d_{ij}\in\mathbb{R}$ for $(i,j)\in\hat{\mathcal{E}}$, $i$ or $j\in\mathcal{S}$. Note that QCQP (\ref{qcqp}) can actually describe ASNL in arbitrary dimensional space.

\begin{remark}
In (\ref{qcqp}), all the angles in a sensor network are taken into account for localization. From the example in Fig. \ref{fig example for asnl}, only partial angles are required to determine locations of unknown sensors. That is, (\ref{qcqp}) contains redundant angle information for localization. Therefore, $\mathcal{T}_{\hat{\mathcal{G}}}$ in (\ref{qcqp}) can be reduced to its subset $\mathcal{T}_{\hat{\mathcal{G}}}^*$. When $(\hat{\mathcal{G}},x)$ is strongly non-degenerate triangulated, \cite[Theorem 7]{Jing18} gives a minimal set of angle constraints for determining angle fixability of $(\hat{\mathcal{G}},x)$, which is also sufficient for network localization. 
\end{remark}

%\begin{remark}
%Each equality constraint in the feasibility problem (\ref{qcqp}) can be handled as penalty terms in the objective function.%, by changing the equality as the square of infeasibility. 
%For example, problem (\ref{qcqp}) can be rewritten as
%\begin{equation}\label{noconstraints}
%\begin{split}
%\min_x &\sum_{(i,j,k)\in\mathcal{T}_{\hat{\mathcal{G}}}}((x_i-x_j)^{\top}(x_i-x_k)-a_{ijk}d_{ij}d_{ik})^2\\
%&+\sum_{(i,j)\in\hat{\mathcal{E}}}(||x_i-x_j||^2-d_{ij}^2)^2\\
%&+\sum_{i\in\mathcal{A}}||x_i-p_i||^2.
%\end{split}
%\end{equation}
%However, the objective in (\ref{noconstraints}) is a fourth order polynomial, which is nonconvex. %For nonconvex optimization, usually only local optimal solution can be obtained. 
%Even when a given ASNL has a unique feasible solution, the local optimal solution to (\ref{noconstraints}) may not correspond to this solution. Instead, every feasible solution to (\ref{qcqp}) corresponds to a class of unknown sensor positions satisfying all the angle constraints. 
%\end{remark}

We make the following assumption for problem (\ref{qcqp}).

\begin{assumption}\label{as feasible}
Problem (\ref{qcqp}) is feasible; there are no sensors overlapping each other; all sensors are static.
\end{assumption}

For an arbitrary sensor network, once the measured angles $a_{ijk}$, $(i,j,k)\in\mathcal{T}_{\hat{\mathcal{G}}}$ are exact, Problem (\ref{qcqp}) is feasible. All the results in this paper will be established on the premise that Assumption \ref{as feasible} is valid. In practice, it may be difficult for measurements to be exact. However, even when only a range for each angle is measured, the feasibility of (\ref{qcqp}) still holds by replacing the first class of equality constraints with inequality constraints. More details are explained in Remark \ref{re advan SDP}.
\subsection{Angle Localizability}

After formulating the ASNL problem as the QCQP (\ref{qcqp}), we strive to answer the question in this subsection: what kind of sensor networks will make the ASNL problem have a unique solution?

\begin{definition}\label{de angle localizability}
	A sensor network is {\it angle localizable} if there is a unique feasible solution to (\ref{qcqp}).
\end{definition}

Fig. \ref{fig localizability} shows three demonstrations for Definition \ref{de angle localizability}. The following Lemma gives a necessary condition for sensor networks to be angle localizable.

\begin{figure}
	\centering
	\includegraphics[width=9cm]{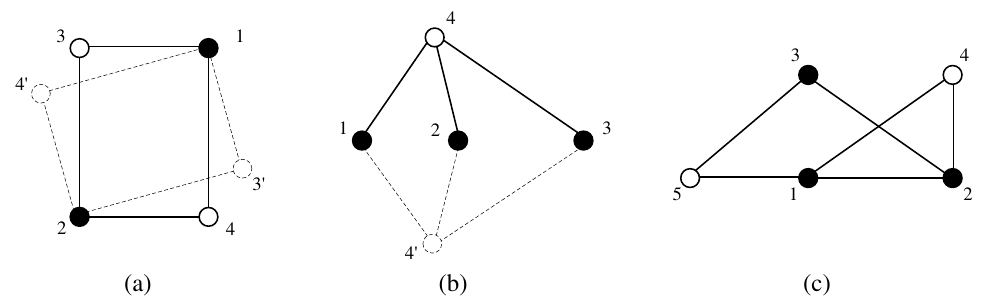}
	\caption{(a) A sensor network that is not angle localizable, sensors 3 and 4 may be incorrectly localized as $3'$ and $4'$. (b) A sensor network that is not angle localizable, sensor 4 may be localized as $4'$. (c) An angle localizable sensor network.}
	\label{fig localizability}
\end{figure}

\begin{lemma} \label{le hyperplane}
	If the sensor network is angle localizable in $\mathbb{R}^2$, then anchors are not all collinear.
\end{lemma}

\begin{proof}
	Suppose that anchors are all collinear. We discuss the following two cases: 

Case 1, all sensors are collinear. Then given any sensor $i\in\mathcal{S}$, there must exist a constant $\delta>0$, such that for $y_i=x_i-\delta\frac{x_i-x_j}{||x_i-x_j||}$, it holds that $\frac{y_i-x_j}{||y_i-x_j||}=\frac{x_i-x_j}{||x_i-x_j||}$ for any $j\in\mathcal{V}$. That is, all the unknown sensors cannot be uniquely localized.

Case 2, not all the sensors are collinear. Let $y\in\mathbb{R}^2$ be a unit vector perpendicular to the line determined by anchors. It can be verified that $(\bar{x}_i-\bar{x}_j)^{\top}(\bar{x}_i-\bar{x}_k)=(x_i-x_j)^{\top}(x_i-x_k)$ for all $(i,j,k)\in\mathcal{T}_{\hat{\mathcal{G}}}$ and $||\bar{x}_i-\bar{x}_j||=||x_i-x_j||$ for all $(i,j)\in\hat{\mathcal{E}}$, where $\bar{x}_i=p_i$ for $i\in\mathcal{A}$, $\bar{x}_i=H_yx_i$ for $i\in\mathcal{S}$, $H_y=I_2-2yy^{\top}$ is the Householder transformation. 
\end{proof}

Lemma \ref{le hyperplane} implies that at least 3 anchors are required for angle localizability of a network in $\mathbb{R}^2$. Note that for SNL based on bearings measured in the global coordinate frame \cite{Zhu14,Zhao16,Zhao18,Trinh18}, the minimum number of anchors required for localizability in $\mathbb{R}^2$ is 2. This is because the angle between two global bearings can be determined as clockwise or counterclockwise in the global coordinate frame. As we claimed before, each angle considered in this paper is totally measured in the local coordinate frame, and cannot be recognized as clockwise or counterclockwise in the global coordinate frame. Fig. \ref{fig localizability} (a) shows a sensor network that is localizable by bearings in the global coordinate frame, but is not angle localizable. In Fig. \ref{fig localizability} (b), the sensor network is not localizable since all anchors are collinear.

The following theorem shows a connection between angle localizability and angle fixability.

\begin{theorem}\label{th localizability=fixability}
	A sensor network $\mathbf{N}=(\hat{\mathcal{G}},x,\mathcal{A})$ is angle localizable in $\mathbb{R}^2$ if and only if $(\hat{\mathcal{G}},x)$ is angle fixable and anchors are not all collinear.
\end{theorem}

\begin{proof}
	Sufficiency. Let $y=(y_1^{\top},...,y_n^{\top})^{\top}\in\mathbb{R}^{2n}$ be a solution to (\ref{qcqp}). Since $(\hat{\mathcal{G}},x)$ is angle fixable, we have $y=c(I_n\otimes\mathscr{R})x+\mathbf{1}_n\otimes\xi$ for some $c\in\mathbb{R}\setminus\{0\}$, $\mathscr{R}\in \text{O}(2)$, $\xi\in\mathbb{R}^2$. Lemma \ref{le hyperplane} implies that there exist at least 3 anchors located non-collinear. Without loss of generality, let $1,2,3$ be the three anchors. Then  $y_i=x_i$, $i=1,2,3$ and $x_1-x_2=y_1-y_2=c\mathscr{R}(x_1-x_2)$. Since $||\mathscr{R}||=1$, we have $|c|=1$. Therefore, $\mathscr{R}'\triangleq c\mathscr{R}\in \text{O}(2)$. Similarly, we have $x_1-x_3=\mathscr{R}'(x_1-x_3)$. Let $A=(x_1-x_2,x_1-x_3)\in\mathbb{R}^{2\times2}$, $A$ must be of full rank. Then we can obtain $\mathscr{R}'=I_2$ from $A=\mathscr{R}'A$. Since $x_1=Y_{12}=\mathscr{R}'x_1+\xi$, we have $\xi=0$. As a result, $y=x$.
	
	Necessity. Lemma \ref{le hyperplane} has shown that anchors are not all collinear, we next prove angle fixability of $(\hat{\mathcal{G}},x)$. Consider $y=(y_1^{\top},...,y_n^{\top})^{\top}\in\mathbb{R}^{2n}$ such that $y\in f^{-1}_{\hat{\mathcal{G}}}(f_{{\hat{\mathcal{G}}}}(x))$, it suffices to prove $y\in\mathscr{S}_x$. Note that the subgraph $\hat{\mathcal{G}}_a$ composed by vertices $\{1,...,n_a\}$ and related edges is complete. Let $x_a=(x_1^{\top},...,x_{n_a}^{\top})^{\top}$, $y_a=(y_1^{\top},...,y_{n_a}^{\top})^{\top}$. From Lemma \ref{le f=gn}, $(\hat{\mathcal{G}}_a,x_a)$ is angle fixable, then $y_a\in\mathscr{S}_{x_a}$. Angle localizability implies that given $y_a$, the rest of coordinates of $y$ such that $f_{\hat{\mathcal{G}}}(x)=f_{\hat{\mathcal{G}}}(y)$ can be uniquely determined. Since a suitable $y$ is an element in $\mathscr{S}_{x}$, hence there must hold $y\in\mathscr{S}_x$.
\end{proof}

Combing Theorem \ref{th angle fixable} and Theorem \ref{th localizability=fixability}, we obtain a graphical condition for angle localizability, see the following  corollary.
\begin{corollary}\label{co}
	A sensor network $\mathbf{N}=(\hat{\mathcal{G}},x,\mathcal{A})$ is angle localizable in $\mathbb{R}^2$ if $(\hat{\mathcal{G}},x)$ has a non-degenerate bilateration ordering, and anchors are not all collinear.
\end{corollary}

Corollary \ref{co} implies that Definition \ref{de af framework} can be used to construct angle localizable sensor networks. Similar approaches for generating bearing localizable networks can be found in \cite{Trinh18b,Trinh19}.

\begin{remark}
When there are no anchors in the network, i.e., $\mathcal{A}=\varnothing$,  sensors' locations can still be determined uniquely up to rotations, translations, uniform scaling and reflections by solving (\ref{qcqp}). In this scenario, Definition \ref{de angle localizability} can be modified by using $\mathscr{S}_{x^*}$ as the unique solution set of (\ref{qcqp}), where $x^*$ corresponds to actual locations of sensors. Due to absence of anchors, it holds that $\hat{\mathcal{G}}=\mathcal{G}$. The centralized and distributed approaches that will be presented later can also be extended to the anchor-free case. 
\end{remark}

\section{CASNL: An Optimization Perspective}\label{sec centralized ASNL}

In this section, a centralized optimization framework is proposed to solve ASNL. In the CASNL, the information sensed by all sensors (including anchors) will be collected in a central unit, and the ASNL problem will be solved by this central unit. An example to illustrate the centralized approach to ASNL is presented in Fig. \ref{fig centralized approach}, where each sensor has its local coordinate frame and transmits local bearing measurements to the central unit, and inter-sensor communications are not required. In this setting, the bearing-based approaches in \cite{Zhu14,Zhao16,Zhao18,Trinh18,Li19} are not applicable because they require each sensor to either measure the bearing measurements in the global coordinate frame or transform local bearings to global bearings by communicating with neighbors.

The significance of studying CASNL can be generally summarized as follows: (i) CASNL does not require each sensor to communicate with other sensors or perform computation.
(ii) A CASNL framework contributes to privacy preserving: since only angle or bearing information is transmitted to the central unit once, the third party cannot obtain any sensor's position unless it detects all the sensed information and knows anchors' positions. (iii) The CASNL and the DASNL have different objectives in practice, thus have different application scenarios. The goal of CASNL is to obtain all sensors' locations at the central unit, based on which the next task is set in a centralized framework, while in DASNL each sensor localizes itself.

\begin{figure}
	\centering
	\includegraphics[width=7cm]{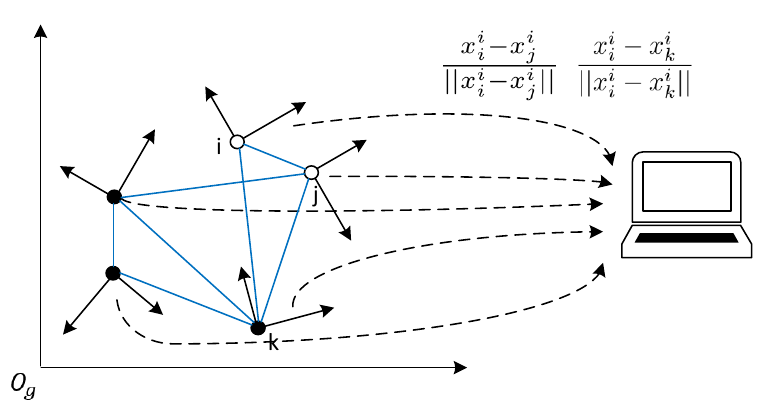}
	\caption{An illustration for CASNL.}
	\label{fig centralized approach}
\end{figure}

\subsection{An SDP Formulation}

To solve the ASNL problem (\ref{qcqp}) from a centralized perspective, we will establish a linear SDP formulation based on ASNL (\ref{qcqp}), and analyze the condition for equivalence of them. Here, ``equivalence" of two problems means that there is a one-to-one correspondence between the solution sets of them. The results in this subsection can be trivially extended to ASNL in higher dimensional space.

Let $$X=(x_{n_a+1},...,x_{n_a+n_s})\in\mathbb{R}^{2\times n_s},$$ $$Y=\begin{pmatrix}
I_2 & X\\
X^{\top} & X^{\top}X
\end{pmatrix}\in\mathbb{R}^{(2+n_s)\times (2+n_s)},$$ $$\tilde{d}=(..., d_{l_{ij}}, ...)^{\top}\in\mathbb{R}^{m},$$ $$D=\tilde{d}\tilde{d}^{\top}\in\mathbb{R}^{m\times m}, m=|\hat{\mathcal{E}}|.$$ Define $e_i\in\mathbb{R}^{n_s}$ and $E_i\in\mathbb{R}^{m}$ as an $n_s$-dimensional and an $m$-dimensional unit vector with the $i$-th entry being 1, respectively. Then ASNL (\ref{qcqp}) is converted into an SDP
\begin{equation}\label{first conversion}
\begin{split}
\min\limits_{Y,D}~~~~~~~~~&~ 0\\
\text{s.t.} ~~(f_i-f_j)^{\top}Y(f_i-f_k)&=a_{ijk}E_{l_{ij}}^{\top}DE_{l_{ik}}, (i,j,k)\in\mathcal{T}_{\hat{\mathcal{G}}}, \\
(f_i-f_j)^{\top}Y(f_i-f_j)&=E_{l_{ij}}^{\top}DE_{l_{ij}}, (i,j)\in\hat{\mathcal{E}},\\
Y_{1:2,1:2}=I_2&,~~Y\succeq 0,  ~~D\succeq 0,\\
\end{split}
\end{equation}
where $Y_{1:2,1:2}$ is the second leading principal submatrix of $Y$, and $f_i\in\mathbb{R}^{n_s+2}$ is defined as
\[
f_i=\left\{
\begin{array}{llll}
(p_i^{\top}, \mathbf{0}_{1\times n_s})^{\top},~~~~~~~i\in\mathcal{A},\\
(\mathbf{0}_{1\times 2},e_{i-n_a}^{\top})^{\top},~~~~~ i\in\mathcal{S}.\\
\end{array}
\right. 
\]

SDP (\ref{first conversion}) is formulated by transforming variables $x$ and $\tilde{d}$ in (\ref{qcqp}) to matrix variables $Y$ and $D$. Since $Y$ and $D$ are uniquely determined by $x$ and $\tilde{d}$, any solution of (\ref{qcqp}) corresponds to a solution of (\ref{first conversion}). However, the converse may not hold. Before we study the condition for equivalence of (\ref{qcqp}) and (\ref{first conversion}), we firstly equivalently transform (\ref{first conversion}) to a standard SDP with only one matrix variable. Denote $$Q_{ijk}=\frac12\bigg[(f_i-f_k)(f_i-f_j)^{\top}+(f_i-f_j)(f_i-f_k)^{\top}\bigg],$$ 
$$Q_{ij}= (f_i-f_j)(f_i-f_j)^{\top},$$ 
$$R_{ijk}=\frac12\bigg[E_{l_{ik}}E_{l_{ij}}^{\top}+E_{l_{ij}}E_{l_{ik}}^{\top}\bigg],~ R_{ij}=E_{l_{ij}}E_{l_{ij}}^{\top},$$ then (\ref{first conversion}) can be equivalently rewritten as
\begin{equation}\label{second conversion}
\begin{split}
\min\limits_{Y,D}&~~~~~~~~~~ 0\\
\text{s.t.} ~~\langle Q_{ijk}, Y \rangle&=a_{ijk}\langle R_{ijk},D\rangle, ~~~~(i,j,k)\in\mathcal{T}_{\hat{\mathcal{G}}}, \\
\langle Q_{ij}, Y \rangle &=\langle R_{ij}, D \rangle, ~~~~~~~~(i,j)\in\hat{\mathcal{E}},\\
Y_{1:2,1:2}&=I_2,~~Y\succeq 0,  ~~D\succeq 0.
\end{split}
\end{equation}
To include both $Y$ and $D$ in one matrix variable, we further define
$$Z=\begin{pmatrix}
Y & \mathbf{0}\\
\mathbf{0} & D
\end{pmatrix}, \Phi_{ijk}=\begin{pmatrix}
Q_{ijk} & \mathbf{0}\\
\mathbf{0}      &   \mathbf{0}_{m\times m}
\end{pmatrix},$$ $$\Psi_{ij}=\begin{pmatrix}
Q_{ij} & \mathbf{0}\\
\mathbf{0}      &   \mathbf{0}_{m\times m}
\end{pmatrix}, \bar{\Psi}_{ij}=\begin{pmatrix}
\mathbf{0}_{(n_s+2)\times (n_s+2)} & \mathbf{0}\\
\mathbf{0}      &   R_{ij}
\end{pmatrix},$$ $$\bar{\Phi}_{ijk}=\begin{pmatrix}
\mathbf{0}_{(n_s+2)\times (n_s+2)} & \mathbf{0}\\
\mathbf{0}      &  R_{ijk}
\end{pmatrix}.$$ Then (\ref{first conversion}) becomes

\begin{equation}\label{SDP Z}
\begin{split}
\min\limits_{Z}~~~~~~ 0~~~~~~~~~~~~&\\
\text{s.t.} ~~\langle \Phi_{ijk}, Z \rangle=a_{ijk}\langle \bar{\Phi}_{ijk},Z\rangle, ~&(i,j,k)\in\mathcal{T}_{\hat{\mathcal{G}}}, \\
\langle \Psi_{ij}, Z \rangle =\langle \bar{\Psi}_{ij}, Z \rangle, ~~(i,j&)\in\hat{\mathcal{E}},\\
Z_{1:n_s+2,n_s+3:n_s+2+m}=\mathbf{0}, & \\
Z_{n_s+3:n_s+2+m,1:n_s+2}=\mathbf{0},&\\
Z_{1:2,1:2}=I_2, Z\succeq 0.~~&
\end{split}
\end{equation}

Given a sensor network, after the angle information measured by sensors is transmitted to the central unit, all the matrices in (\ref{SDP Z}) except variable $Z$ will be known. More specifically, $\mathcal{T}_{\hat{\mathcal{G}}}$ and $\hat{\mathcal{E}}$ are determined by measurements from sensors; $\Phi_{ijk}$ and $\Psi_{ij}$ are determined by $\mathcal{T}_{\hat{\mathcal{G}}}$, $\hat{\mathcal{E}}$ and positions of anchors; $\bar{\Phi}_{ijk}$ and $\bar{\Psi}_{ij}$ are determined by $\mathcal{T}_{\hat{\mathcal{G}}}$ and $\hat{\mathcal{E}}$, respectively; $a_{ijk}$ is the angle information based on measurements from sensors. 

In what follows, we will analyze the property of solutions of (\ref{SDP Z}), and propose a condition for equivalence of (\ref{qcqp}) and (\ref{SDP Z}).

\begin{lemma}\label{le rank Z}
	Let $Z$ be a solution to (\ref{SDP Z}), then $\rank(Z)\geq 3$.
\end{lemma}
\begin{proof}
	From the definition of $Z$, we have $\rank(Z)=\rank(Y)+\rank(D)$. Since $Y$ is symmetric, and $Y_{1:2,1:2}=I_2$, we have $Y=\begin{pmatrix}
	I_2 & Y_{12}\\
	Y_{12}^{\top} & Y_{22}
	\end{pmatrix}$ with $Y_{12}\in\mathbb{R}^{2\times n_s}$ and $Y_{22}\in\mathbb{R}^{n_s\times n_s}$. Then $\rank(Y)=\rank(I_2)+\rank(Y_{22}-Y_{12}^{\top}Y_{12})\geq 2$. For matrix $D$, Assumption \ref{as feasible} implies that $D_{l_{ij}l_{ik}}$ is not a zero matrix. Then $\rank(D)\geq1$. In conclusion, $\rank(Z)\geq 3$.
\end{proof}

From $Z\succeq0$, we obtain that $Y_{22}\succeq Y_{12}^{\top}Y_{12}$, and $D$ is nontrivial. Since there is no rank constraint for $Z$ in (\ref{SDP Z}), it is possible that the rank of a solution $Z$ is greater than $3$. As a result, a solution of (\ref{SDP Z}) may not correspond to a solution to (\ref{qcqp}). The following lemma shows that by imposing a rank constraint on $Z$, SDP (\ref{SDP Z}) becomes equivalent to (\ref{qcqp}).

\begin{lemma}\label{le d+1 to equivalence}
	If the rank of every solution to (\ref{SDP Z}) is $3$, then (\ref{SDP Z}) is equivalent to (\ref{qcqp}).
\end{lemma}

\begin{proof}
	Given $X$ as a solution to (\ref{qcqp}), it is obvious that the induced $Z$ is a solution to (\ref{SDP Z}). Next we prove that given a solution $Z$ (with rank $3$) to (\ref{SDP Z}), we can find a unique solution $X$ corresponding to $Z$ such that $X$ is a solution to (\ref{qcqp}). Again, consider $Y=\begin{pmatrix}
	I_2 & Y_{12}\\
	Y_{12}^{\top} & Y_{22}
	\end{pmatrix}$ with $Y_{12}\in\mathbb{R}^{2\times n_s}$ and $Y_{22}\in\mathbb{R}^{n_s\times n_s}$. Since $\rank(Z)=3$, from the proof of Lemma \ref{le rank Z}, we have $\rank(Y)=2$, implying $\rank(Y_{22}-Y_{12}^{\top}Y_{12})=0$, then $Y_{22}=Y_{12}^{\top}Y_{12}$. Moreover, $\rank(D)=1$, then there exists some vector $\bar{d}\in\mathbb{R}^{|\hat{\mathcal{E}}|}$ such that $D=\bar{d}\bar{d}^{\top}$. Since $Z$ and $D$ are in the desired forms, the constraints in (\ref{SDP Z}) are equivalent to those in (\ref{qcqp}). Let $X=Y_{12}$, $X$ must be the solution to (\ref{qcqp}).
\end{proof}

Lemma \ref{le d+1 to equivalence} implies that ASNL (\ref{qcqp}) can be equivalently formulated as  the SDP (\ref{SDP Z}) with an rank constraint $\rank(Z)=3$, which is a non-convex optimization and generally NP-hard. Although extensive methods for rank-constrained optimization have been proposed in the literature, e.g., \cite{Wan19,Sun17,You19}, they only guarantee local convergence for general cases. That is, the initial guess given to the algorithm needs to be sufficiently close to an optimal solution. Moreover, solving a rank-constrained optimization is usually time-consuming especially when the problem is of large size. In the next subsection, by utilizing some inherent properties of ASNL, we will derive a condition for removing the rank constraint on the unknown matrix $Z$.

\begin{remark}\label{re advan SDP}
By incorporating the rank condition  in Lemma \ref{le d+1 to equivalence}, the SDP formulation (\ref{SDP Z}) is equivalent to the original ASNL and is scalable to noises and bounded unknown disturbances. For example, in specific scenarios, each angle constraint is obtained within a range due to the existence of noises or disturbances, i.e., $a_{ijk}\in[l_{ijk},u_{ijk}],(i,j,k)\in\mathcal{T}_{\hat{\mathcal{G}}}$. In this case, the first class of equality constraints can be revised as inequality constraints
$$l_{ijk}\langle \bar{\Phi}_{ijk},Z\rangle\leq\langle \Phi_{ijk}, Z \rangle\leq u_{ijk}\langle \bar{\Phi}_{ijk},Z\rangle, ~(i,j,k)\in\mathcal{T}_{\hat{\mathcal{G}}}.$$
Under this setting, when the sensing graph has a sufficient amount of edges, SDP (\ref{SDP Z}) can still be solved with high accuracy. An example (Example \ref{ex disturbance}) will be given to demonstrate this fact in the simulation section. Moreover, the discussions on SDP (\ref{SDP Z}) in the following two subsections are applicable to the case with bounded unknown disturbances as well.
\end{remark}

\subsection{Relation to SDP Relaxation}

Since a rank-constrained  SDP is difficult to solve, we consider how to relax the rank-constrained SDP to a convex problem (i.e., remove the rank constraint). In the literature, e.g., \cite{Biswas04,Biswas06,Simonetto14}, SDP relaxation has been widely used to solve RSNL. However, usually a solution to the relaxed problem may not correspond to a solution to the original problem. In this subsection, we will establish connections between the relaxed formulation in (\ref{SDP Z}) and the original nonconvex QCQP in (\ref{qcqp}).

The following theorem shows that under a specific graph condition, the rank of $Z$ can be efficiently constrained once the rank of $D$ is constrained.
\begin{theorem}\label{th rank}
Let $Z=\left(\begin{smallmatrix}
Y & \mathbf{0}\\
\mathbf{0} & D
\end{smallmatrix}\right)$ be a solution to (\ref{SDP Z}) and $\rank(D)=1$, then $\rank(Z)=3$ must hold if and only if

(i) anchors are not all collinear;

(ii) $(\hat{\mathcal{G}},x)$ is angle fixable in $\mathbb{R}^2$ and its angle fixability is invariant to space dimensions.
\end{theorem}
\begin{proof}
	See Appendix.
\end{proof}

%Note that condition (i) in Theorem \ref{th rank} cannot be neglected as a sufficiency condition. If condition (i) does not hold, even if $(\hat{\mathcal{G}},x)$ is angle fixable, positions of sensors may not be unique, and there may exist some solution to (\ref{SDP Z}) with rank greater than $3$.

We say a triangulated framework $(\mathcal{G},p)$ is {\it acute-triangulated} if each triangle in this framework only contains acute angles, i.e., for any $(i,j),(i,k),(j,k)\in\mathcal{E}$, it holds that $g_{ij}^{\top}g_{ik}$, $g_{ji}^{\top}g_{jk}$, $g_{ki}^{\top}g_{kj}\in(0,1)$. As a result, an acute-triangulated framework must be strongly non-degenerate, and has a non-degenerate bilateration ordering. Next we give a graph condition for constraining the rank of $D$.

\begin{theorem}\label{th D rank1}
Let $Z=\left(\begin{smallmatrix}
Y & \mathbf{0}\\
\mathbf{0} & D
\end{smallmatrix}\right)$ be a solution to (\ref{SDP Z}). If $(\hat{\mathcal{G}},x)$ is acute-triangulated, then $\rank(D)=1$.
\end{theorem}
\begin{proof}
See Appendix.
\end{proof}

\begin{remark}
In simulation experiments, by solving the SDP formulation in (\ref{SDP Z}), all unknown sensors can always be correctly localized when $(\hat{\mathcal{G}},x)$ is strongly non-degenerate triangulated. We will make further efforts to prove this in future. However, when $(\hat{\mathcal{G}},x)$ has a non-degenerate bilateration ordering but is not triangulated, the solution to (\ref{SDP Z}) may correspond to incorrect localization results. An example (Example \ref{counterexample}) will be given in Section \ref{sec simulation}. 
\end{remark}

By virtues of Theorems \ref{th rank}, \ref{th D rank1} and Lemma \ref{le universal fixability}, the following result is derived.

\begin{theorem}\label{th equivalence}
Given a sensor network $\mathbf{N}=(\hat{\mathcal{G}},x,\mathcal{A})$, if $(\hat{\mathcal{G}},x)$ contains an acute-triangulated subframework, and anchors are not all collinear, then\\
	(i) (\ref{qcqp}), (\ref{first conversion}), (\ref{second conversion}) and (\ref{SDP Z}) are all equivalent;\\
    (ii) (\ref{SDP Z}) has a unique solution with rank 3.
\end{theorem}
   
The conclusions in Theorem \ref{th equivalence} imply that ASNL is angle localizable, and QCQP (\ref{qcqp}) is equivalent to its SDP relaxation, thus can be efficiently solved within polynomial time.

To clearly demonstrate the results stated in Theorems \ref{th rank}, \ref{th D rank1} and \ref{th equivalence}, we summarize the relationships between different conditions for $(\hat{G},x)$ and relaxation results in Fig. \ref{fig relax relationships}. Here we assume that anchors are always not all collinear.
\begin{figure}
	\centering
	\includegraphics[width=7cm]{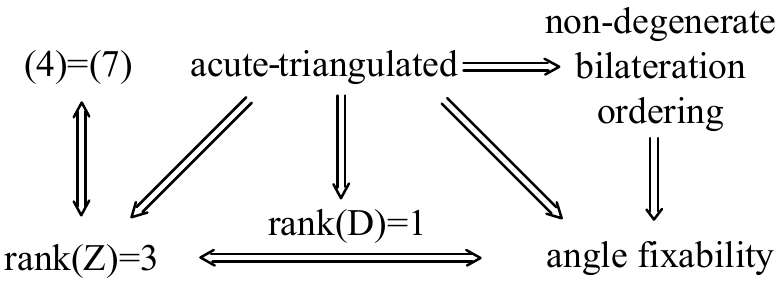}
	\caption{The relationships between different conditions for the grounded framework and relaxation results. Anchors are always considered to be not all collinear.}
	\label{fig relax relationships}
\end{figure}

\subsection{Decomposition for Large-Scale ASNL}

The formulation in (\ref{SDP Z}) is a standard linear SDP, thereby can be globally solved by the interior-point methods in polynomial time. However, when the dimension of $Z$ is very large, due to high computational costs of considering the positive semi-definite constraint, existing interior-point algorithms may be unable to find the solution within reasonable computational time. Note that large-scale networks are ubiquitously encountered in practice. From intuitive observation, when the sensor network is of large scale, matrices $\Phi_{ijk}$, $\bar{\Phi}_{ijk}$, $\Psi_{ij}$ and $\bar{\Psi}_{ij}$ in (\ref{SDP Z}) are usually large and sparse. In this subsection, we will recognize some special features of ASNL and transform a large and sparse ASNL to a linear SDP with multiple semi-definite cone constraints for smaller-sized matrices. %Detailed proofs for the theorems in this subsection can be found in the full version of this paper (\cite{Jing19arxiv}).

It is observed that compared to (\ref{SDP Z}), the problem in (\ref{second conversion}) has a smaller size and fewer constraints. Therefore, we will focus on problem (\ref{second conversion}) directly. By Theorem \ref{th rank}, the original ASNL (\ref{qcqp}) is equivalent to the following SDP:
\begin{equation}\label{YD optimization}
\begin{split}
\min\limits_{Y,D}&~~~~~~~~~~ 0\\
\text{s.t.} ~~\langle A_i, Y \rangle+&\langle B_i,D\rangle=c_i, ~~~~i=1, ..., s,\\
Y,D\succeq &0,~\rank(D)=1,\\
%\rank(D)=1,
\end{split}
\end{equation}
where $A_i\in\mathbb{R}^{(n_s+2)\times(n_s+2)}$, $B_i\in\mathbb{R}^{m\times m}$, $s=|\mathcal{T}_{\hat{\mathcal{G}}}|+|\hat{\mathcal{E}}|+4$. 

\subsubsection{Decomposition for $Y$}
Due to the definition of the inner product for matrices, an entry in $Y$, e.g., $Y_{ij}$, is constrained by some equality constraint if and only if the entry in the same position of some $A_i$ is nonzero. Let $A=\sum_{i=1}^s\abs(A_i)$, where $\abs(A_i)$ is the element-wise absolute operation. Then $Y_{ij}$ is constrained by some equality constraint if and only if $A_{ij}\neq0$. Here we use $E(A)$ to denote the sparsity pattern of $A$, where
$$
E(A)=\{(i,j)\in V(A)\times V(A):A_{ij}\neq0,i\neq j\},
$$
$V(A)=\{1,...,n_s+2\}$. Denote graph $\mathcal{G}(A)=(V(A),E(A))$, we have the following result.

\begin{theorem}\label{th A is chordal}
If $\hat{\mathcal{G}}$ has a bilateration ordering and for any $i\in\mathcal{A}$ connecting to some sensor $k\in\mathcal{S}$, there exists another anchor $j\in\mathcal{A}$ such that $(p_i-p_j)_x(p_i-p_j)_y\neq0$, then graph $\mathcal{G}(A)$ is chordal.
\end{theorem}
\begin{proof}
	See Appendix.
\end{proof}

The condition in Theorem \ref{th A is chordal} implies that the edges between several pairs of anchors are not parallel to both $x$-axis and $y$-axis of the global coordinate frame. Note that $\{q\in\mathbb{R}^{2n_a}:(q_i-q_j)_x(q_i-q_j)_y=0,\,i,j\in\{1,...,n_a\}\}$ is of measure zero. That is, a randomly generated network satisfies the condition in Theorem \ref{th A is chordal} with probability 1. 

When graph $\mathcal{G}(A)$ is chordal, by Lemma \ref{le chordal decomposition}, the constraint $Y\succeq0$ in (\ref{YD optimization}) can be replaced by positive semidefinite constraints $Y_i=Q_{\mathcal{C}_i}YQ_{\mathcal{C}_i}^{\top}\succeq0$, where $Y_i\in\mathbb{R}^{|\mathcal{C}_i|\times|\mathcal{C}_i|}$, $\mathcal{C}_i$ is the set of vertices corresponding to the $i$-th maximal clique of $\mathcal{G}(A)$. However, if sensor $k\in\mathcal{S}$ has no anchor neighbors, it must hold that $A_{1,k'}=A_{2,k'}=0$, where $k'=k-n_a+2$. As a result, $Y_{1k'}$ and $Y_{2k'}$ are not constrained in the converted optimization problem. In the ASNL problem, we hope to find $X=Y_{1:2,3:n_s+2}=(x_{n_a+1},...,x_{n_a+n_s})$, which contains position information of all unknown sensors. If $Y_{1k'}$ and $Y_{2k'}$ are not constrained, then we cannot obtain the correct position of sensor $k$ by solving the decomposed optimization problem directly. In \cite{Zheng19}, to obtain the solution to the original undecomposed problem, a positive semi-definite matrix completion problem\footnote{A matrix completion problem is to recover missing entries of a matrix from a set of known entries \cite{Grone84}.} should be addressed. In this paper, to avoid solving the matrix completion problem, we extend graph $\mathcal{G}(A)$ by adding edges such that positions of all unknown sensors can be constrained. 

Let $\bar{A}=A+\left(\begin{smallmatrix}
\mathbf{0} & \mathbf{1}_2\mathbf{1}_{n_s}^{\top}\\
\mathbf{0} & \mathbf{0}
\end{smallmatrix}\right)$, and decompose matrix $Y$ according to the sparsity pattern $E(\bar{A})$. Then $Y_{1:2,3:n_s+2}$ is always constrained. By following similar lines to proofs of Theorem \ref{th A is chordal}, the following result can be obtained.

\begin{lemma}
If $\hat{\mathcal{G}}$ has a bilateration ordering, then $\mathcal{G}(\bar{A})=(\mathcal{V}(\bar{A}),\mathcal{E}(\bar{A})$ is chordal.
\end{lemma}

Together with Lemma \ref{le chordal decomposition}, we have the following decomposition law.

\begin{theorem}\label{th decomposeY}
If $\hat{\mathcal{G}}$ has a bilateration ordering, then $Y\succeq0$ is equivalent to $Y_i=Q_{\mathcal{C}_i(\bar{A})}YQ_{\mathcal{C}_i(\bar{A})}^{\top}\succeq0$, where $\mathcal{C}_i(\bar{A})$ is the set of vertices corresponding to the $i$-th maximal clique of graph $\mathcal{G}(\bar{A})$.
\end{theorem}

\subsubsection{Decomposition for $D$}

Similar to $\mathcal{G}(A)$, we can obtain graph $\mathcal{G}(B)=(V(B),E(B))$, where $V(B)=\{1,...,m\}$, $E(B)$ is the aggregate sparsity pattern for $B_i$, $i=1,...,s$, $B=\sum_{i=1}^s\abs(B_i)$. Unlike $\mathcal{G}(A)$, graph $\mathcal{G}(B)$ is more sparse and can never be chordal. By observing the form of (\ref{second conversion}), one can see that only partial elements of $D$ are constrained. Although the desired $D$ should be of rank 1, since our final goal is to find $Y$, we only require all the constrained elements of $D$ to satisfy the rank 1 constraint.

\begin{theorem}\label{th decomposeD}
Suppose that $(\hat{\mathcal{G}},x)$ is acute-triangulated, the solution $Y$ to (\ref{YD optimization}) remains invariant if constraints ``$D\succeq0$" and ``$\rank(D)=1$" are relaxed to ``$D_i=Q_{\mathcal{C}_i(B)}DQ_{\mathcal{C}_i(B)}^{\top}\succeq0$", 
where $\mathcal{C}_i(B)$ is the set of vertices corresponding to the $i$-th maximal clique of graph $\mathcal{G}(B)$.
\end{theorem}
\begin{proof}
	See Appendix.
\end{proof}

\subsubsection{Decomposed ASNL}

Combining Theorems \ref{th decomposeY} and \ref{th decomposeD}, we obtain the following result.

\begin{theorem}\label{th acute decomposition}
If $(\hat{\mathcal{G}},x)$ is acute-triangulated, then the ASNL problem in (\ref{qcqp}) is equivalent to the following optimization:
\begin{equation}\label{YiDi optimization}
\begin{split}
\min\limits_{Y,D}&~~~~~~~~~~ 0\\
\text{s.t.} ~~\langle A_i,& Y \rangle+\langle B_i,D\rangle=c_i, ~~~~i=1, ..., s,\\
Y_i=&Q_{\mathcal{C}_i(\bar{A})}YQ_{\mathcal{C}_i(\bar{A})}^{\top},~Y_i\succeq 0,~~ i=1, ..., \xi,\\
D_i=&Q_{\mathcal{C}_i(B)}YQ_{\mathcal{C}_i(B)}^{\top},~D_i\succeq 0,~~ i=1, ... \zeta, \\
%Y_i&\succeq 0, i=1, ..., \xi,\\
%D_i&\succeq 0, i=1, ... \zeta
\end{split}
\end{equation}
where $A_i$, $B_i$ and $s$ are the same as those in (\ref{YD optimization}), $\xi$ is the number of maximal cliques of $\mathcal{G}(\bar{A})$, and $\zeta$ is the number of maximal cliques of $\mathcal{G}(B)$.
\end{theorem}

\begin{comment}
Let $\text{vec}(X)\in\mathbb{R}^{\sigma^2}$ denote the vector stacking columns of matrix $X\in\mathbb{R}^{\sigma\times\sigma}$ in the order that they appear in $X$, and $\mat(x)$ be a matrix such that $x=\text{vec}(\mat(x))$.
By denoting $$\mathscr{A}=(\text{vec}(A_1),...,\text{vec}(A_s))^{\top}\in\mathbb{R}^{s\times (n_s+2)^2},$$ $$\mathscr{B}=(\text{vec}(B_1),...,\text{vec}(B_s))^{\top}\in\mathbb{R}^{s\times m^2},$$ $$\mathscr{M}=[\mathscr{A}, \mathscr{B}]\in\mathbb{R}^{s\times[(n_s+2)^2+m^2]},$$  $$y=\begin{pmatrix}
\text{vec}(Y)\\
\text{vec}(D)
\end{pmatrix}\in\mathbb{R}^{(n_s+2)^2+m^2},$$ $$z_i=\text{vec}(Y_i)\in\mathbb{R}^{|\mathcal{C}_i(\bar{A})|^2},i=1,...,\xi$$ $$J_i=[Q_{\mathcal{C}_i(\bar{A})}\otimes Q_{\mathcal{C}_i(\bar{A})},\mathbf{0}_{|\mathcal{C}_i(\bar{A})|^2\times m^2}],i=1,...,\xi,$$ $$z_{\xi+i}=\text{vec}(D_i)\in\mathbb{R}^9,i=1,...,\zeta$$
$$J_{\xi+i}=[\mathbf{0}_{9\times(n_s+2)^2}, Q_{\mathcal{C}_i(B)}\otimes Q_{\mathcal{C}_i(B)}], i=1,...,\zeta$$ $$c=(c_1,...,c_s)^{\top},$$ problem (\ref{YiDi optimization}) can be rewritten as
\begin{equation}\label{vector optimization}
\begin{split}
&\min\limits_{y,z_1,...,z_{\xi+\zeta}}~~ 0\\
\text{s.t.} ~~&~~~~~~~~\mathscr{M}y=c, \\
&z_i=J_iy, i=1, ..., \xi+\zeta,\\
&\text{mat}(z_i)\succeq 0, i=1, ..., \xi+\zeta.\\
\end{split}
\end{equation}
\end{comment}
The SDP formulation in (\ref{YiDi optimization}) has been studied in the literature \cite{Sun14,Kalbat15,Zheng19}. In \cite{Kalbat15}, a distributed algorithm was proposed. In \cite{Zheng19}, an algorithm based on fast alternating direction method of multiplies (ADMM) was developed, which can be implemented distributively as well. In practice, the convergence speed for solving (\ref{YiDi optimization}) may depend on the sensing graph of a specific network. 

\begin{remark}
Theorems \ref{th decomposeY} and \ref{th decomposeD} indicate that the condition for decomposing $D$ is more demanding than that for decomposing $Y$. When $(\hat{\mathcal{G}},x)$ has a bilateration ordering but is not acute-triangulated, we can decompose the positive semi-definite constraint on $Y$ according to Theorem \ref{th decomposeY}, and decompose constraints on $D$ via the approach in \cite{You19}. In \cite{Miller19}, the authors proposed another approach for solving the rank-constrained optimization via chordal decomposition. By extending graph $\mathcal{G}(B)=(V(B),E(B))$ to a chordal graph, the approach in \cite{Miller19} can be implemented to solve ASNL.
\end{remark}

%\begin{remark}
%
%\end{remark}

\subsection{ASNL in a Noisy Environment}

In practice, measurements obtained by sensors are usually inexact. In this subsection, we study ASNL in the presence of stochastic noises. 

When the angle measurements contain noises, we replace $a_{ijk}$ by $\bar{a}_{ijk}=a_{ijk}+\mathbf{n}_{ijk}$, where $a_{ijk}$ is the actual cosine value of the angle between $x_i-x_j$ and $x_i-x_k$, $\mathbf{n}_{ijk}$ denotes the measurement noise effect. Now assume that we only have $\bar{a}_{ijk}$ available, $a_{ijk}$ is an unknown variable to be determined. Let $a=(...,a_{ijk},...)^{\top}\in\mathbb{R}^{|\mathcal{T}_{\hat{\mathcal{G}}}|}$ and $\bar{a}=(...,\bar{a}_{ijk},...)^{\top}\in\mathbb{R}^{|\mathcal{T}_{\hat{\mathcal{G}}}|}$. Inspired by \cite{Simonetto14}, we model the ASNL with noise as the following likelihood maximization problem,
\begin{equation}\label{ML}
\begin{split}
\min_{x,a,\tilde{d}}~ f(a)\\
\text{s.t.}, ~~(x_i-x_j)^{\top}(x_i-x_k)&=a_{ijk}d_{ij}d_{ik},  (i,j,k)\in\mathcal{T}_{\hat{\mathcal{G}}}\\
||x_i-x_j||^2&=d_{ij}^2, ~~~~~~~~~~(i,j)\in\hat{\mathcal{E}}\\
x_i&=p_i, ~~~~~~~~~~~~~~~i\in\mathcal{A}
\end{split}
\end{equation}
where $$f(a)=-\sum_{(i,j,k)\in\mathcal{T}_{\hat{\mathcal{G}}}}\ln \mathbf{P}_{ijk}(a_{ijk}|\bar{a}_{ijk}),$$
$\mathbf{P}_{ijk}(a_{ijk}|\bar{a}_{ijk})$ is the sensing probability density function, which depends on the property of noise $\mathbf{n}_{ijk}$. When $\mathbf{P}_{ijk}(a_{ijk}|\bar{a}_{ijk})$ is a log-concave function of $a_{ijk}$, $f(a)$ is always convex. Now we simply consider the Gaussian zero-mean white noise, i.e., $\mathbf{n}_{ijk}\sim N(0,\sigma_{ijk}^2)$, the objective function becomes 
\begin{equation}\label{f(a)}
f(a)=\sum_{(i,j,k)\in\mathcal{T}_{\hat{\mathcal{G}}}}\frac{(a_{ijk}-\bar{a}_{ijk})^2}{\sigma_{ijk}^2}.
\end{equation}
Note that $(\ref{ML})$ is no longer a QCQP since $a_{ijk}$ becomes a variable. However, by introducing new variables $\mathbf{d}_{ijk}$ and constraints $\mathbf{d}_{ijk}=d_{ij}d_{ik}$, (\ref{ML}) can be converted to a QCQP again.

To convert (\ref{ML}) into an SDP with a reasonable scale, we introduce new $3\times3$ matrix variables $\Lambda_{l_{ijk}}=\lambda_{l_{ijk}}\lambda_{l_{ijk}}^{\top}$, where $\lambda_{l_{ijk}}=(a_{ijk},\mathbf{d}_{ijk},1)^{\top}\in\mathbb{R}^{3}$, similar to (\ref{YD optimization}), the noisy ASNL is equivalent to the following SDP,
\begin{equation}\label{ML SDP}
\begin{split}
\min_{Y,D,\Lambda_i} \sum_{i=1}^{|\mathcal{T}_{\hat{\mathcal{G}}}|}\langle F_i&(\bar{a}),\Lambda_i \rangle\\
\text{s.t.} ~~~~\langle A'_i,Y\rangle+\langle B'_i,D\rangle+&\sum_{j=1}^{|\mathcal{T}_{\hat{\mathcal{G}}}|}\langle C'_i,\Lambda_j\rangle=c'_i, i=1,...,s',\\
Y,D\succeq0, &~\rank(D)=1,\\
\Lambda_j(3,3)=1, &~~\Lambda_j\succeq0,\\
\rank(\Lambda_j)=1, &~~j=1,...,|\mathcal{T}_{\hat{\mathcal{G}}}|,
\end{split}
\end{equation} 
here $F_i(\bar{a})$ is determined by $f(a)$ in (\ref{f(a)}), $\Lambda_j(3,3)$ represents the element in the 3rd row and 3rd column of $\Lambda_j$, $Y\in\mathbb{R}^{(n_s+2)\times(n_s+2)}$ and $D\in\mathbb{R}^{m\times m}$ are in the same sense as those in (\ref{YD optimization}), but $A'_i\in\mathbb{R}^{(n_s+2)\times(n_s+2)}$, $B'_i\in\mathbb{R}^{m\times m}$, $c'_i\in\mathbb{R}$ and $s'$ are different from $A_i$, $B_i$, $c_i$ and $s$ in (\ref{YD optimization}). More specifically, $s'=2|\mathcal{T}_{\hat{\mathcal{G}}}|+m+4$.

Let $A'=\sum_{i=1}^s\abs(A'_i)$, $B'=\sum_{i=1}^s\abs(B'_i)$, we observe that the sparsity patterns of $A'$ and $B'$ are the same as those of $A$ and $B$ in (\ref{YD optimization}), thus semi-definite constraints for $Y$ and $D$ in (\ref{ML SDP}) can still be decomposed in the way described in the last subsection. Also the rank constraint for $D$ can be removed when $(\hat{\mathcal{G}},x)$ is acute-triangulated. Efficient algorithms for solving (\ref{ML SDP}) can be found in \cite{You19,Miller19}. If we simply ignore rank constraints, the resulting ASNL relaxation is a linear SDP, which can be solved by an SDP solver, e.g., CVX \cite{Grant09}, directly.

\section{Distributed ASNL via Inter-Sensor Communications}\label{sec distributed ASNL}

Solving ASNL in a centralized manner requires all sensors to transmit information to a unified central unit, which generates high computation and communication load in practice. Although the algorithms in \cite{You19,Miller19,Kalbat15,Zheng19} can solve decomposed ASNL in a distributed fashion, all the required data should be collected in a central unit beforehand. Moreover, the algorithms in \cite{You19,Miller19,Kalbat15,Zheng19} cannot be distributively executed by assigning each subtask to a sensor node.

In this section, we propose a distributed algorithm for ASNL, where each sensor computes its own position by using only local information obtained from its neighbors. The sensing measurements between neighboring sensors are still relative bearings in their own local coordinate frames. Similar to most of the existing distributed optimization references, we assume that each sensor is able to communicate with its neighbors. Note that it would be impossible for a sensor to localize itself if it has no access to exact positions of neighbors.

Distributed bearing-based localization algorithms in \cite{Trinh18,Li19} can also be implemented when relative bearing measurements are measured in local coordinate frames. However, they require the sensors to cooperatively obtain bearing measurements in a unified coordinate frame via frequent inter-sensor communications. In contrast, our distributed protocol only requires finite time inter-sensor communications, thereby saves a significant amount of communication costs.

\subsection{Bilateration Localization}\label{subsec: BL}

Given an angle localizable sensor network $(\hat{\mathcal{G}},x,\mathcal{A})$, in which all the sensors have been localized. Now we show that after placing a new sensor $k$ being a common neighbor of $i$ and $j$ in the network such that $x_i-x_k$ and $x_j-x_k$ are not collinear, $x_k$ can be uniquely determined by two angles subtended at $i$ and two angles subtended at $j$. An example is shown in Fig. \ref{determinek}. Note that since the sensor network is angle localizable, $(\hat{\mathcal{G}},x)$ is angle fixable. Then both $i$ and $j$ must have two neighboring nodes not lying collinear. Without loss of generality, let $i_1$ and $i_2$ be the two neighbors of $i$, $j_1$ and $j_2$ be the two neighbors of $j$ (It is possible that $i_1$ or $i_2=j$, $j_1$ or $j_2=i$). Since sensors $i$, $i_1$ and $i_2$ have already been localized, the bearings $g_{ii_1}$ and $g_{ii_2}$ with respect to the global coordinate system can both be obtained. In addition, both $\cos\angle1$ and $\cos\angle2$ can be computed by sensor $i$ using bearings $g^i_{ii_1}$, $g^i_{ii_2}$ and $g^i_{ik}$ measured in its local coordinate frame. Let $g_{ik}$ be the bearing between $i$ and $k$ in the global coordinate frame, then we have
$$
g_{ii_1}^{\top}g_{ik}=g_{ii_1}^{iT}g_{ik}^i, ~~~~
g_{ii_2}^{\top}g_{ik}=g_{ii_2}^{iT}g_{ik}^i.
$$
Recall that $g_{ii_1}$ and $g_{ii_2}$ are not collinear, $g_{ik}$ can be uniquely solved. For simplicity, we denote $\mathcal{F}_g$ as the function to compute $g_{ik}$, i.e., 
\begin{equation}\label{Fg}
 g_{ik}=\mathcal{F}_g(g_{ii_1},g_{ii_2},g^i_{ii_1},g^i_{ii_2},g^i_{ik}). 
\end{equation}
Similarly, $g_{jk}$ can be obtained. It is observed that the following equations must hold
$$
\det(x_i-x_k~~ g_{ik})=0,~~~~ 
\det(x_j-x_k~~ g_{jk})=0.$$
Since $g_{ik}$ and $g_{jk}$ are linearly independent, $x_k$ can be uniquely solved. We denote $\mathcal{F}_x$ as the function to compute $x_k$, i.e., 
\begin{equation}\label{Fx}
x_k=\mathcal{F}_x(x_i,x_j,g_{ik},g_{jk}).
\end{equation}

\begin{figure}
	\centering
	\includegraphics[width=4cm]{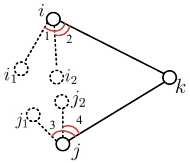}
	\caption{An example for localizing a sensor via bilateration localization.}
	\label{determinek}
\end{figure}

\subsection{A Distributed Protocol for ASNL}

To utilize the bilateration localization method in a distributed manner, we assign the tasks of solving $g_{ik}$ and $g_{jk}$ via (\ref{Fg}) to the localized sensors $i$ and $j$, respectively; and assign the task of solving (\ref{Fx}) to the unlocalized sensor $k$. We consider that each sensor has two modes: localized and unlocalized. Each anchor is in the localized mode. Only the localized sensors transmit information to their neighbors, while all sensors are always able to sense relative bearings from neighbors. As a result, each sensor is able to determine if a neighbor is in the localized mode by checking if it receives information from this neighbor. Now we propose a distributed protocol called ``Bilateration Localization Protocol (BLP)". The pseudo codes of BLP are shown in Protocol 1.

Fig. \ref{fig BLP} illustrates the procedure of localizing a sensor network by implementing BLP. The black nodes denote anchors, red nodes are sensors in localized mode, white sensors are in unlocalized mode. It is shown that all the sensors are localized at step 3. Note that the graph in Fig. \ref{fig BLP} is not only the sensing graph $\mathcal{G}$, but also the grounded graph $\hat{\mathcal{G}}$. It is important to note that if there are no links between anchors in sensing graph $\mathcal{G}$, graph $\hat{\mathcal{G}}$ remains the same but BLP is not applicable because each anchor is not able to measure relative bearings from other anchors. Hence, a condition for sensing graph $\mathcal{G}$ is required to guarantee the validity of Protocol 1.

\begin{figure}[t]
	\centering
	\includegraphics[width=7cm]{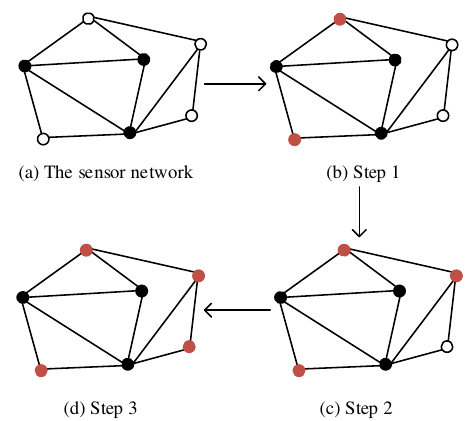}
	\caption{The localization procedure via BLP. Here the sensing graph $\mathcal{G}$ is identical to the grounded graph $\hat{\mathcal{G}}$.}
	\label{fig BLP}
\end{figure}

	Let $\mathcal{L}(t)$ and $\mathcal{U}(t)$ be the sets of sensors in localized and unlocalized mode, respectively, before BLP is implemented at step $t$, $t=0,1,2...$ . 

		Define 
\begin{multline*}
	\mathcal{L}^*(t)=\{i\in\mathcal{L}(t): \det(g_{ii_1}, g_{ii_2})\neq0~\text{for}~\text{some}\\~i_1,i_2\in\mathcal{L}(t), (i,i_1), (i,i_2)\in\mathcal{E}.\}
\end{multline*}

\begin{protocol}{The bilateration localization protocol for ASNL}\label{BLP}
 Each sensor has two modes: localized and unlocalized. Each sensor is able to transmit and receive information to/from neighbors, and sense relative bearings from neighbors in its local coordinate frame.
	
	\textbf{Sensor $i$ in the localized mode:}
	
	\textit{Available information:} Position $x_i$, position $x_j$ received from localized neighbor $j$, bearings $g^i_{ij}$, $j\in\mathcal{N}_i$ sensed from neighbors. Denote $\mathcal{N}_{il}$ and $\mathcal{N}_{iu}$ as the sets of localized and unlocalized neighbors of $i$, respectively.
	
	\textit{Protocol:}
	\begin{enumerate}
		\item \textbf{for all} $k\in\mathcal{N}_{iu}$ \textbf{do}
		\item Arbitrarily choose distinct $i_1$ and $i_2$ from $\mathcal{N}_{il}$ such that $x_i-x_{i_1}$ and $x_i-x_{i_2}$ are not collinear
		\item Compute $g_{ik}=\mathcal{F}_g(g_{ii_1},g_{ii_2},g^i_{ii_1},g^i_{ii_2},g^i_{ik})$ by solving the linear equations in (\ref{Fg})
		\item Transmit $x_i$, $g_{ik}$ to sensor $k$
		\item \textbf{end for}
		\item \textbf{for all} $k\in\mathcal{N}_{il}$ \textbf{do}    
		\item Transmit $x_i$ to sensor $k$
		\item \textbf{end for}
	\end{enumerate}	
	\textbf{Sensor $k$ in the unlocalized mode:} 
	
	\textit{Available information:} Positions $x_i$ and bearings $g_{ik}$ received from localized neighbors $i\in\mathcal{N}_{kl}$, bearings $g^k_{ik}$ sensed from neighbors $i\in\mathcal{N}_k$ in its local coordinate frame.
	
	\textit{Protocol:}
	\begin{enumerate}
		\item \textbf{If} positions from more than two neighbors received and these positions are not collinear with $x_k$ \textbf{then}
		\item Arbitrarily choose distinct $i$ and $j$ from $\mathcal{N}_{kl}$ such that $x_i-x_k$ and $x_j-x_k$ are not collinear
		\item Compute $x_k=\mathcal{F}_x(x_i,x_j,g_{ik},g_{jk})$ by solving the linear equations in (\ref{Fx})
		\item Switch to localized mode
		\item \textbf{end if}
	\end{enumerate}			
\end{protocol}

Let $\mathcal{N}_i$ be the neighbor set of sensor $i$ in the grounded graph $\hat{\mathcal{G}}$. Using the bilateration localization method in Subsection \ref{subsec: BL}, the set of sensors that will be localized at step $t$ is  
\begin{multline*}
	\Delta (t)=\{k\in\mathcal{N}_i\cap\mathcal{N}_j\cap\mathcal{U}(t): i,j\in\mathcal{L}^*(t), \\
	\text{det}(x_k-x_i, x_k-x_j)\neq0\}.
\end{multline*}

If $(\mathcal{G},x)$ has a non-degenerate bilateration ordering, and a subframework of $(\mathcal{G},x)$ with vertices in a subset of $\mathcal{L}(t)$ has a non-degenerate bilateration ordering, it always holds that $\Delta(t)\neq\varnothing$ when $\mathcal{U}(t)\neq\varnothing$. Since $|\mathcal{U}(0)|=n_s$ is finite, $\mathcal{U}(t)$ converges to a zero set in finite time. The convergence speed depends on the volume of $\Delta(t)$ at each time step $t$, which is determined by the  grounded framework. Based on the above analysis, we present the following result.

\begin{theorem}\label{th BLP}
	If $(\mathcal{G},x)$ has a non-degenerate bilateration ordering, and there exists a subframework $(\mathcal{G}_l,x_l)$ of $(\mathcal{G},x)$ containing anchors only has a non-degenerate bilateration ordering, then Protocol 1 solves ASNL within $n_s$ steps.
\end{theorem}

If the sensors can be labelled such that $\{1,...,n_a\}$ is the set of anchors, for any $i>n_a$, the $i$-th vertex has exactly two neighbors with one of them being the $(i-1)$-th vertex, BLP solves ASNL by $n_s$ steps. In practice, usually the number of steps for convergence is smaller than $n_s$ because $|\Delta(t)|$ is greater than 1 for some steps.

\begin{remark}\label{re BLP accuracy}
Observe that when implementing BLP, the accuracy of localizing an unknown sensor depends on the accuracy of the information received from its localized neighbors. If the neighbors of a sensor are inaccurately localized, then this sensor will be inaccurately localized accordingly. As a result, when the sensor network is in a noisy environment, the position estimation errors will accumulate during the implementation of BLP. The later a sensor is localized, the greater error its estimated position has. In conclusion, although BLP accomplishes the localization task with a fast speed, it requires high accuracy of sensed measurements. The topic of how to design a more scalable distributed localization protocol is one of our ongoing research endeavors.
\end{remark}

\begin{remark}
When all the angle constraints are accurately obtained, the bilateration localization approach is applicable to the CASNL problem by simulating sensors' behaviors in the central unit, which has a faster speed than solving any SDP in Section \ref{sec centralized ASNL}. Moreover, since all the anchors' information can be utilized, Protocol 1 is valid as long as the network is angle localizable. The advantages of using the SDP formulation for CASNL have been explained in Section \ref{sec centralized ASNL}. 
\end{remark}

\section{Simulation Examples}\label{sec simulation}

In this section, we present four simulation examples. The first two examples show that the equivalence between ASNL (\ref{qcqp}) and the decomposed linear SDP (\ref{YiDi optimization}) holds if the grounded framework $(\hat{\mathcal{G}},x)$ is acute-triangulated, but may not hold when $(\hat{\mathcal{G}},x)$ has a non-degenerate bilateration ordering. The third case shows the ASNL solution considering noisy measurements. The fourth example demonstrates Theorem \ref{th BLP} and shows that Protocol 1 has a fast speed. The last example compares the centralized method and the distributed method for an ASNL problem with disturbed measurements. All simulation examples are run in Matlab environments using a standard desktop.

\subsection{Simulations for CASNL}

\subsubsection{Noise-Free ASNL}
\begin{example}\label{ex triangulated}
Consider a sensor network $(\hat{\mathcal{G}},x,\mathcal{A})$ with $n=30$ sensors and $n_a=3$ anchors among them randomly distributed in the unit box $[0,1]^2$, and $(\hat{\mathcal{G}},x)$ is acute-triangulated. The sensor network is shown in Fig. \ref{fig sensor network} (a). By solving the decomposed SDP (\ref{YiDi optimization}) via CVX/SeDuMi \cite{Grant09}, we obtain the results plotted in Fig. \ref{fig sensor network} (b) with the computational time being 4.4853s. It is observed that the locations of unknown sensors estimated by CVX/SeDuMi closely match the real locations, which is consistent with Theorem \ref{th equivalence} and Theorem \ref{th acute decomposition}. It is worth noting that when we solve the undecomposed SDP (\ref{SDP Z}), each sensor can still be correctly localized. But the computational time is 51.4574s. Hence, the proposed decomposition method significantly improves the computational speed.

\begin{figure}
	\centering
	\includegraphics[width=9cm]{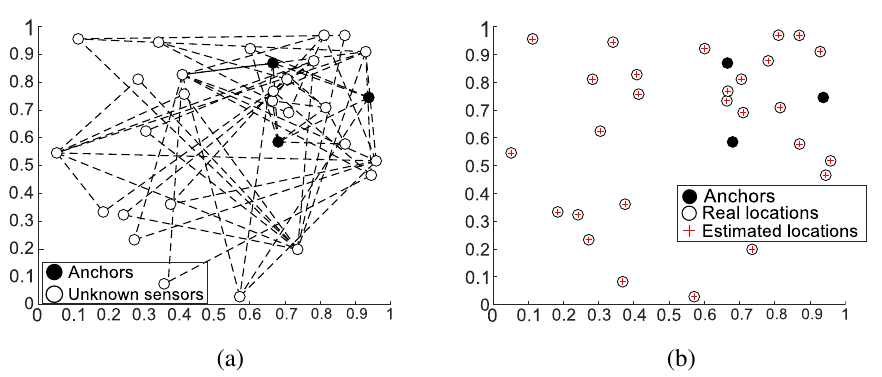}
	\caption{(a) An acute-triangulated sensor network. (b) The locations estimated by CVX/SeDuMi almost perfectly match the real locations.}
	\label{fig sensor network}
\end{figure}
\end{example}

\begin{example}\label{counterexample}
Consider a sensor network $(\hat{\mathcal{G}},x,\mathcal{A})$ randomly distributed in the unit box $[0,1]^2$, $(\hat{\mathcal{G}},x)$ has a bilateration ordering but is not triangulated, thus is angle localizable. The grounded framework is shown in Fig. \ref{fig bi sensor network} (a). The localization results obtained by solving (\ref{YiDi optimization}) via CVX/SeDuMi are depicted in Fig. \ref{fig bi sensor network} (b), from which we observe that not all unknown sensors can be localized. This is because the graphical conditions in both Theorem \ref{th equivalence} and Theorem \ref{th acute decomposition} are not satisfied. The incorrect localization result means that the solution matrix $D$ either is not positive semi-definite or has a rank greater than 1. After checking the solution, we find that $D$ is not positive semi-definite. We also tried to solve the undecomposed SDP (\ref{YD optimization}) for this example, the resulting localization results are still incorrect. This is due to the invalidity of the condition in Theorem \ref{th equivalence}, which makes the rank of $D$ greater than 1. 

\begin{figure}
	\centering
	\includegraphics[width=9cm]{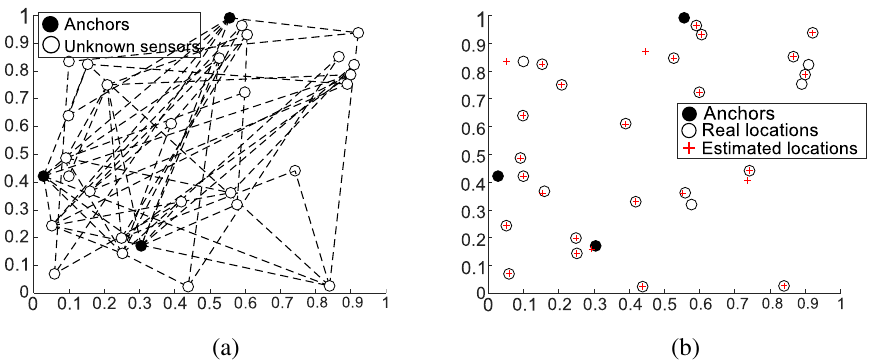}
	\caption{(a) A sensor network with a bilateration ordering but is not triangulated. (b) Several unknown sensors are incorrectly localized by solving the decomposed SDP.}
	\label{fig bi sensor network}
\end{figure}
\end{example}

\subsubsection{Noisy ASNL}
\begin{example}\label{ex noisy ASNL}
Consider a sensor network $(\hat{\mathcal{G}},x,\mathcal{A})$ with 3 anchors and 5 unknown sensors randomly distributed in the unit box $[0,1]^2$ (as shown in Fig. \ref{fig noisyASNL} (a)) suffering a Gaussian white noise, and $(\hat{\mathcal{G}},x)$ is acute-triangulated. Then the rank constraint on $D$ in (\ref{ML SDP}) can be removed. An additive zero-mean white noise with a uniform standard deviation $\sigma$ is applied to each angle measurement, i.e., $\bar{a}_{ijk}=a_{ijk}+\mathbf{n}_{ijk}$, $\mathbf{n}_{ijk}\sim N(0,\sigma^2)$. Similar to \cite{Sun17,You19}, we solve (\ref{ML SDP}) by an iterative rank minimization approach, which is to solve a series of linear SDPs as follows:

\begin{equation}\label{noisy iteration}
\begin{split}
\min_{Y,D,\Lambda_i^l,r_l}~~ \sum_{i=1}^{|\mathcal{T}_{\hat{\mathcal{G}}}|}\langle F_i(\bar{a}),\Lambda_i^l \rangle&+w_lr_l\\
\text{s.t.} ~~~~\langle A'_i,Y\rangle+\langle B'_i,D\rangle+\sum_{j=1}^{|\mathcal{T}_{\hat{\mathcal{G}}}|}\langle C'_i,&\Lambda_j^l\rangle=c'_i, i=1,...,s',\\
Y,D\succeq0, \\
r_lI_2-V_j^{lT}\Lambda_j^lV_j^l\succeq 0,\\
\Lambda_j^l(3,3)=1, ~~\Lambda_j^l\succeq0, &~~~~j=1,...,|\mathcal{T}_{\hat{\mathcal{G}}}|,
\end{split}
\end{equation} 
where $w_l=\alpha^l w_0$ is set as an increasing positive sequence, i.e., $\alpha>1$, $w_0>0$, $V_j^l=(v_{j1}^{l-1},v_{j2}^{l-1})^{\top}\in\mathbb{R}^{2\times3}$, $v_{j1}$ and $v_{j2}$ are two eigenvectors corresponding to the two smallest eigenvalues of $\Lambda_j^{l-1}$, which is obtained by solving the SDP formulation in (\ref{noisy iteration}) at step $l-1$. The initial state of each $\Lambda_j$, i.e., $\Lambda_j^0$, is obtained by solving (\ref{ML SDP}) without considering the rank constraints.

We solve a sequence of SDPs (\ref{noisy iteration}) by CVX/SeDuMi successively until $r_l<\epsilon$ at some step $l^*$, where $\epsilon$ is a positive scalar close to 0. The solution $(Y, D, \Lambda_j^l)$ to (\ref{noisy iteration}) at step $l^*$ is regarded as the solution to (\ref{ML SDP}). Note that the selections of $w_0$ and $\alpha$ are quite important for convergence of the iterative rank minimization algorithm. In \cite{You19}, a convolutional neural network (CNN) is designed to seek appropriate $w_0$ and $\alpha$.

Now we consider $\sigma=0.005$, and set $w_0=1$, $\alpha=1.3$, by solving noisy ASNL with random white noise 100 times, the localization results are depicted in Fig. \ref{fig noisyASNL} (b). We observe that the unknown sensor whose two neighbors are both anchors can be localized with a small error, while the unknown sensor with two different types of neighboring sensors is localized with a relatively larger error. 

\begin{figure}
	\centering
	\includegraphics[width=9cm]{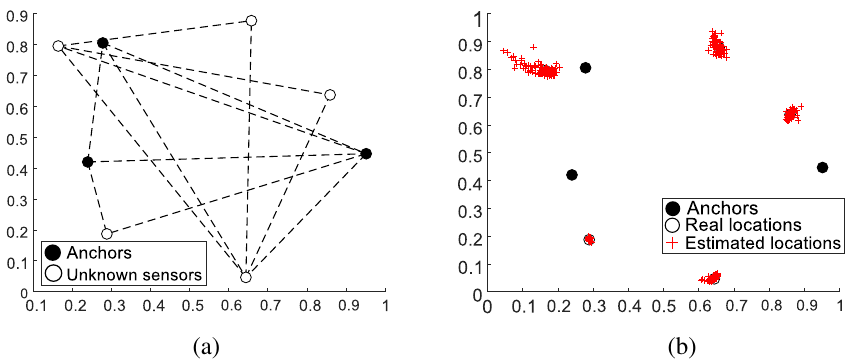}
	\caption{(a) An acute-triangulated sensor network in a noisy environment. (b) The locations of unknown sensors estimated by CVX/SeDuMi are close to their actual locations.}
	\label{fig noisyASNL}
\end{figure}
\end{example}

\subsection{Simulations for DASNL}
\begin{example}
Consider three sensor networks with 100, 500, 1000 sensors in the plane, positions of sensors are randomly generated by Matlab such that each network has a non-degenerate bilateration ordering. Moreover, each network has only 3 anchor nodes among all the sensors. By implementing the distributed protocol BLP, the three ASNL problems are solved, respectively. Fig. \ref{fig distributed} (a) shows evolution of the percentage of unlocalized sensors with respect to all unknown sensors. We observe that as the network size grows, the required number of iterations increases slowly. Fig. \ref{fig distributed} (b) depicts the history of the volume of sensors localized along each step. It is shown that during the implementation of BLP, the number of sensors localized per step increases at the beginning, and usually decreases sharply after half of total iteration steps.
	
\begin{figure}
	\centering
	\includegraphics[width=9cm]{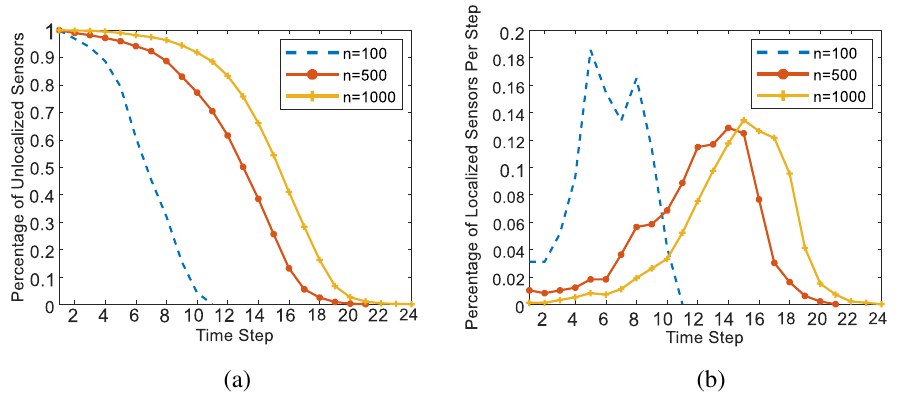}
	\caption{(a) Evolution of the percentage of unlocalized sensors. (b) Evolution of the percentage of localized sensors per step.}
	\label{fig distributed}
\end{figure}

We further tested 10 randomly generated examples with 100, 500 and 1000 sensors, respectively. In each example, there are only 3 anchors and the framework has a non-degenerate bilateration ordering. The average computational time, the average convergence step, the average time for a single step, as well as the average computational error for each case are shown in Table. \ref{tabDASNL}. Here ``CT" denotes ``Computational Time", ``CS" denotes ``Convergence Step", ``CTPS" means ``Computational Time Per Step". The computational error is computed by $\sqrt{\sum_{i=n_a+1}^{n}||x_i^*-x_i^e||^2}$, where $x_i^*$ and $x_i^e$ are the actual location and the estimated location of sensor $i$, respectively. It is observed that BLP always solves ASNL within $n_s$ steps, which is consistent with Theorem \ref{th BLP}. We also tested an example with 100 sensors where the grounded framework is acute-triangulated by solving SDP (\ref{YiDi optimization}). The computational time is 176.5689s. Therefore, the computational speed of BLP  is much faster than the centralized approach. 

\begin{table}[htbp]		
	\centering
		\fontsize{10}{10}\selectfont
		\begin{threeparttable}
			\caption{DASNL via BLP}
			\label{tabDASNL}
			\begin{tabular}{ccccc}
				\toprule			
				$n$  &    CT(sec)& CS & CTPS(sec)&Error\\
				\midrule
				100& 0.0132  &  13.9 & 0.0009&6.6961e-11\\
				500&   0.0811   & 20.8  & 0.0039&1.7837e-9\\
				1000& 0.1575 &23& 0.0068&2.4118e-8\\
				\bottomrule
			\end{tabular}
		\end{threeparttable}
\end{table}
\end{example}

\subsection{ASNL with Bounded Disturbances}

Although the distributed protocol has a fast speed, it requires high accuracy of the measurements. In this subsection, we solve ASNL with disturbances on measurements by the centralized approach and the distributed approach, respectively. Different from Example \ref{ex noisy ASNL}, the disturbances considered here are bounded. We will show that the centralized approach is more robust to unknown disturbances compared with the distributed approach.

\begin{example}\label{ex disturbance}
	
Consider a sensor network $(\hat{\mathcal{G}},x,\mathcal{A})$ with 3 anchors and 7 unknown sensors randomly distributed in the unit box $[0,1]^2$. The framework $(\mathcal{G},x)$ is considered to be acute-triangulated. Motivated by \cite{Shames12}, the disturbed local bearing measurement between sensors $i$ and $j$ measured at sensor $i$ can be written as $\bar{g}_{ij}^i=g_{ij}^i+\tau_{ij}$, where $\tau_{ij}$ is the unknown error and $||\tau_{ij}||\leq 0.01$. Then each angle constraint becomes $\bar{a}_{ijk}=a_{ijk}+\tau_{ijk}$, where $\tau_{ijk}=\tau_{ij}^Tg_{ij}+\tau_{ik}^Tg_{ik}+\tau_{ij}^T\tau_{ik}\in[-0.0201,0.0201]$. To solve ASNL via the centralized approach, we replace the equality constraints involving angles in (\ref{SDP Z}) by the following inequality constraints:
\[
\begin{split}
\langle\Phi_{ijk},Z\rangle&\geq(\bar{a}_{ijk}-0.0201)\langle\bar{\Phi}_{ijk},Z\rangle,\\
 \langle\Phi_{ijk},Z\rangle&\leq(\bar{a}_{ijk}+0.0201)\langle\bar{\Phi}_{ijk},Z\rangle, ~~(i,j,k)\in\mathcal{T}_{\hat{\mathcal{G}}}. 
\end{split}
\]
Then the actual positions of sensors must correspond to a feasible solution to (\ref{SDP Z}) with inequality constraints. When using the distributed protocol,  $\bar{g}_{ij}^i$ and $\bar{a}_{ijk}$ are directly employed as the local bearing and the angle constraint for each sensor $i$. 

In Fig. \ref{fig disturbance}, for a network satisfying conditions in both Theorems \ref{th equivalence} and \ref{th BLP}, the localization results obtained by the centralized and the distributed methods are shown, respectively. The results for both cases are obtained within $0.02$s. It is observed that the centralized approach still has high precision, but the distributed approach has a large estimation error. %After testing more examples, it is shown that when the sensing graph is dense enough, the centralized approach always captures much higher accuracy than the distributed approach. This is consistent with Remark \ref{re BLP accuracy}.

\begin{figure}
	\centering
	\includegraphics[width=9cm]{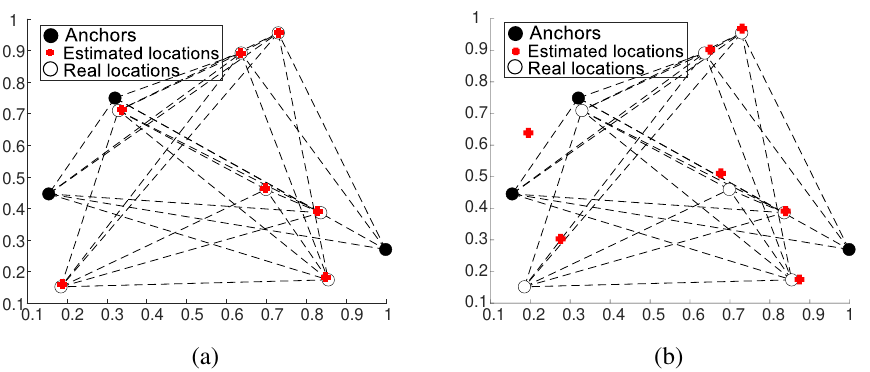}
	\caption{(a) Localization results obtained by solving (\ref{SDP Z}) with inequality constraints. (b) Localization results obtained by implementing Protocol 1.}
	\label{fig disturbance}
\end{figure}

\end{example}

\section{CONCLUSIONS}\label{sec conclusion}

This paper presented comprehensive analysis for angle-based sensor network localization (ASNL).  A notion termed angle fixability was proposed to recognize frameworks that can be uniquely determined by angles up to translations, rotations, reflections and uniform scaling. It has been proved that any framework with a non-degenerate bilateration ordering is angle fixable. The ASNL problem was shown to have a unique solution if and only if the grounded framework is angle fixable, and has been solved in centralized and distributed approaches, respectively. The CASNL in a noise-free environment was modeled as a rank-constrained SDP, which is proved to be equivalent to a linear SDP when the grounded framework is acute-triangulated. A decomposition strategy was proposed to efficiently solve large-scale ASNL problems. The CASNL in a noisy environment was studied via a maximum likelihood formulation, and was also formulated as an SDP with multiple rank constraints and semi-definite constraints. Distributed ASNL was realized by using a bilateration localization approach based on inter-sensor communications.

\section{Appendix: Proofs for Theorems \ref{th rank}, \ref{th D rank1},  \ref{th A is chordal}, \ref{th decomposeD}}

Proof of Theorem \ref{th rank}:	Sufficiency. Suppose that there is a solution $\tilde{Z}$ such that $\rank(\tilde{Z})>3$. Consider $\tilde{Y}=\left(\begin{smallmatrix}
I_2 & Y_{12}\\
Y_{12}^{\top} & Y_{22}
\end{smallmatrix}\right)$ as a part of the solution $\tilde{Z}$. Then there must hold $Y_{22}\succeq Y_{12}^{\top}Y_{12}$ and $Y_{22}\neq Y_{12}^{\top}Y_{12}$. Hence, there exists some nontrivial $Y'_{12}\in\mathbb{R}^{r\times n_s}$ such that $Y_{22}=Y_{12}^{\top}Y_{12}+Y_{12}'^{\top}Y_{12}'$. Note that given anchors' locations $P=(p_1,...,p_{n_a})\in\mathbb{R}^{2\times n_a}$, if $Y_{12}\in\mathbb{R}^{2\times n_s}$ is a feasible set of locations for sensors, then given anchors' locations $(P^{\top},\mathbf{0}_{n_a\times r})^{\top}\in\mathbb{R}^{(2+r)\times n_a}$ in $\mathbb{R}^{2+r}$, $(Y_{12}^{\top},Y_{12}'^{\top})^{\top}\in\mathbb{R}^{(2+r)\times n_s}$ is also a feasible set of locations for sensors. Note that $(Y_{12}^{\top},\mathbf{0}_{n_s\times r})^{\top}\in\mathbb{R}^{(2+r)\times n_s}$ is also a solution. Let $\hat{p}=(p_1^{\top},...,p_{n_a}^{\top})^{\top}$, $\hat{x}=(\hat{x}_1^{\top},...,\hat{x}_{n_s}^{\top})^{\top}\in\mathbb{R}^{2n_s}$, $\bar{x}=(\bar{x}_1^{\top},...,\bar{x}_{n_s}^{\top})^{\top}\in\mathbb{R}^{2n_s}$, $\hat{x}_i$ be the $(i+n_a)$-th column of $(Y_{12}^{\top},\mathbf{0}_{n_s\times r})^{\top}$ and $\bar{x}_i$ be the $(i+n_a)$-th column of $(Y_{12}^{\top},Y_{12}'^{\top})^{\top}$ for $i\in\mathcal{S}$. The condition $\rank(D)=1$ implies that $(\hat{p}^{\top}, \hat{x}^{\top})^{\top}$ and $(\hat{p}^{\top},\bar{x}^{\top})^{\top}$ are two different feasible realizations of framework $(\hat{\mathcal{G}},x)$. Since anchors are not all collinear, $\hat{p}$ is non-degenerate. Then  $(\hat{p}^{\top}, \hat{x}^{\top})^{\top}$ can never be obtained by a trivial motion from $(\hat{p}^{\top},\bar{x}^{\top})^{\top}$. That is, angle fixability of $(\hat{\mathcal{G}},x)$ is not preserved in $\mathbb{R}^{2+r}$, which is a contradiction. 

Necessity. We first prove that $(\hat{\mathcal{G}},x)$ is angle fixable in $\mathbb{R}^2$. Due to Theorem \ref{th localizability=fixability}, it suffices to show that (\ref{qcqp}) has a unique solution. Suppose this is not true, by Lemma \ref{le d+1 to equivalence}, (\ref{SDP Z}) also has multiple solutions. Let $$Z_1=\begin{pmatrix}
Y_1 & ~\\
~ & D_1
\end{pmatrix}, Z_2=\begin{pmatrix}
Y_2 & ~\\
~ & D_2
\end{pmatrix}$$ be two different solutions to (\ref{SDP Z}), where $$Y_1=\begin{pmatrix}
I_2 & X_1\\
X_1^{\top} & X_1^{\top}X_1
\end{pmatrix}, Y_2=\begin{pmatrix}
I_2 & X_2\\
X_2^{\top} & X_2^{\top}X_2
\end{pmatrix},$$ then we conclude that $Z_3=\frac12 Z_1+\frac12Z_2$ is also a solution to (\ref{SDP Z}). As a result, $$\frac12 Y_1+\frac12Y_2=\begin{pmatrix}
I_2 & \frac12 X_1+\frac12 X_2\\
\frac12 X_1^{\top}+\frac12 X_2^{\top} & \frac12 X_1^{\top}X_1+\frac12 X_2^{\top}X_2
\end{pmatrix}.$$ Since $Z_3$ is a solution to (\ref{SDP Z}), $$\frac12 X_1^{\top}X_1+\frac12 X_2^{\top}X_2=(\frac12 X_1+\frac12 X_2)^{\top}(\frac12 X_1+\frac12 X_2).$$ It follows that $||X_1-X_2||=0$. Since $\rank(D)=1$ and all the diagonal elements of $D$ can be determined by $X$, $D$ is uniquely determined by $X$. Then we have $D_1=D_2$. Accordingly, $Z_1=Z_2$, which is a contradiction. Hence, $(\hat{\mathcal{G}},x)$ is angle fixable. By Lemma \ref{le hyperplane}, anchors must be not all collinear.

To show that the angle fixability of $(\hat{\mathcal{G}},x)$ is invariant to space dimensions, we note that from the proof of sufficiency, if $(\hat{\mathcal{G}},x)$ is not angle fixable in $\mathbb{R}^{2+r}$, we can always accordingly find a solution to (\ref{qcqp}) with rank $3+r$. Hence the proof is completed.
\QEDA

To prove Theorem \ref{th D rank1}, the following lemma will be used.

\begin{lemma}\label{le M_{23}}
	Consider a positive semi-definite matrix $M\in\mathbb{R}^{3\times3}$ with positive diagonal entries and one missing non-diagonal entry. If each $2\times2$ principal submatrix associated with available elements is of rank 1, then $M$ is uniquely completable.  
\end{lemma}
\begin{proof}
	Without loss of generality, let $M_{23}$ be the missing entry, then $M_1=\left(\begin{smallmatrix}
	M_{11} & M_{12}\\
	M_{12} & M_{22}
	\end{smallmatrix}\right)$ and $M_2=\left(\begin{smallmatrix}
	M_{11} & M_{13}\\
	M_{13} & M_{33}
\end{smallmatrix}\right)$ are both positive semi-definite and of rank 1. Suppose that $M_1=(a~ b)^{\top}(a~ b)$, $M_2=(a~ c)^{\top}(a~ c)$. As a result, $M=\left(\begin{smallmatrix}
	a^2&ab&ac\\
	ab &b^2&M_{23}\\
	ac&M_{23}&c^2
	\end{smallmatrix}\right)$. Since $M$ is positive semi-definite, we have $\det(M)\geq0$. Then we can derive that $a^2(M_{23}^2-2bcM_{23}+b^2c^2)\leq0$. Together with $a^2\neq0$, we have $M_{23}=bc$.
\end{proof}

Proof of Theorem \ref{th D rank1}: Without loss of generality, suppose $Y=\left(\begin{smallmatrix}
I_2 & X\\
X^{\top} & \bar{X}^{\top}\bar{X}
\end{smallmatrix}\right)$, where $\bar{X}=(\bar{x}_1,...,\bar{x}_{n_s+2})\in\mathbb{R}^{(2+r)\times (n_s+2)}$, $r\geq0$ is an integer. It follows from (\ref{second conversion}) that
$$(\bar{x}_i-\bar{x}_j)^{\top}(\bar{x}_i-\bar{x}_k)=a_{ijk}D_{l_{ij}l_{ik}}, (i,j,k)\in\mathcal{T}_{\hat{\mathcal{G}}},$$
$$||\bar{x}_i-\bar{x}_j||^2=D_{l_{ij}l_{ij}}, (i,j)\in\hat{\mathcal{E}}.$$
For any $(i,j,k)\in\mathcal{T}_{\hat{\mathcal{G}}}$ such that $(j,k)\in\hat{\mathcal{E}}$, since the angle between $(i,j)$ and $(i,k)$ is acute, $a_{ijk}>0$. From $D\succeq0$, we have $D_{l_{ij}l_{ik}}^2\leq D_{l_{ij}l_{ij}}D_{l_{ik}l_{ik}}$. Then
$$(\bar{x}_i-\bar{x}_j)^{\top}(\bar{x}_i-\bar{x}_k)\leq a_{ijk}||\bar{x}_i-\bar{x}_j||||\bar{x}_i-\bar{x}_k||.$$
Let $\theta_1$ be the angle between $\bar{x}_i-\bar{x}_j$ and $\bar{x}_i-\bar{x}_k$. Then $\theta_1\geq\arccos a_{ijk}$. Similarly, let $\theta_2$ and $\theta_3$ be the angles between $\bar{x}_j-\bar{x}_i$ and $\bar{x}_j-\bar{x}_k$, $\bar{x}_k-\bar{x}_i$ and $\bar{x}_k-\bar{x}_j$, respectively. It holds that $\theta_2\geq\arccos a_{jik}$ and $\theta_3\geq\arccos a_{kij}$. Note that $\arccos a_{ijk}+\arccos a_{jik}+\arccos a_{kij}=\pi$, and $\theta_1+\theta_2+\theta_3=\pi$, it follows that $\theta_1=\arccos a_{ijk}$, $\theta_2=\arccos a_{jik}$ and $\theta_3=\arccos a_{kij}$. As a result, $D_{l_{ij}l_{ik}}^2= D_{l_{ij}l_{ij}}D_{l_{ik}l_{ik}}$.

By Lemma \ref{le M_{23}}, if $(i,j,k)\in\mathcal{T}_{\hat{\mathcal{G}}}$, $(j,k)\in\hat{\mathcal{E}}$ and $(i,j,h)\in\mathcal{T}_{\hat{\mathcal{G}}}$, $(j,h)\in\hat{\mathcal{E}}$, then $D_{l_{ik}l_{ih}}^2= D_{l_{ik}l_{ik}}D_{l_{ih}l_{ih}}$. Without loss of generality, let $k<h$, $l_{ij}<l_{ik}<l_{ih}$, $y=D_{l_{ik}l_{ih}}$. From $D\succeq0$, we have 
$$\det\begin{pmatrix}
D_{l_{ij}l_{ij}} & D_{l_{ij}l_{ik}} & D_{l_{ij}l_{ih}}\\
D_{l_{ij}l_{ik}} & D_{l_{ik}l_{ik}} & y\\
D_{l_{ij}l_{ih}} &      y           & D_{l_{ih}l_{ih}}  
\end{pmatrix}\geq0.$$
Together with $D_{l_{ij}l_{ik}}^2= D_{l_{ij}l_{ij}}D_{l_{ik}l_{ik}}$ and $D_{l_{ij}l_{ih}}^2= D_{l_{ij}l_{ij}}D_{l_{ih}l_{ih}}$, we can derive that $y=D_{l_{ik}l_{ik}}D_{l_{ih}l_{ih}}$.

By Lemma \ref{le M_{23}}, we can obtain that for any three edges in the graph, e.g., $(i,j)$, $(k,h)$ and $(u,v)$, if $D_{l_{ij}l_{kh}}^2=D_{l_{ij}l_{ij}}D_{l_{kh}l_{kh}}>0$ and $D_{l_{kh}l_{uv}}^2=D_{l_{kh}l_{kh}}D_{l_{uv}l_{uv}}>0$, then there must hold that $D_{l_{ij}l_{uv}}^2=D_{l_{ij}l_{ij}}D_{l_{uv}l_{uv}}>0$. Since $(\hat{\mathcal{G}},x)$ is triangulated, and the anchors are not all collinear, we have $D_{ij}^2=D_{ii}D_{jj}>0$ for all $i,j\in\{1,...,m\}$. That is, $\rank(D)=1$.
\QEDA

Proof of Theorem \ref{th A is chordal}: We will show that if $(i,j),(i,k)\in E(A)$ and $j\neq k$, then $(j,k)\in E(A)$, i.e., $A_{jk}>0$. Note that $i,j,k\in\{1,...,n_s+2\}$ and 

\[
\begin{split}
A=&\frac12\sum_{(i,j,k)\in\mathcal{T}_{\hat{\mathcal{G}}}}\bigg|\left[(f_i-f_k)(f_i-f_j)^{\top}+(f_i-f_j)(f_i-f_k)^{\top}\right]\bigg|\\
&+\sum_{(i,j)\in\hat{\mathcal{E}}}\bigg|(f_i-f_j)(f_i-f_j)^{\top}\bigg|.
\end{split}	
\] 
Without loss of generality, we consider the following cases:

Case 1. $i,j\in\{1,2\}$, $k>2$. Let $k'=k-2+n_a$, then $k'\in\mathcal{S}$. Note that $A_{ik}\neq0$ only if there exists at least one anchor $i'$ such that $(i',k')\in\mathcal{E}$. Let $j'$ be another anchor distinct to $i'$ such that $(p_{i'}-p_{j'})_x(p_{i'}-p_{j'})_y\neq0$, then 
\[
\begin{split}
M&=(f_{i'}-f_{j'})(f_{i'}-f_{k'})^{\top}=\begin{pmatrix}
p_{i'}-p_{j'}\\\mathbf{0}_{n_s\times1}
\end{pmatrix}(p_{i'}^{\top},-e_{k-2}^{\top})\\
&=\begin{pmatrix}
(p_{i'}-p_{j'})p_{i'}^{\top} & -(p_{i'}-p_{j'})e_{k-2}^{\top}\\
\mathbf{0}_{n_s\times2} & \mathbf{0}_{n_s\times n_s}
\end{pmatrix}.
\end{split}
\] It can be computed that $M_{jk}=(p_{i'}-p_{j'})_x$ if $j=1$ and $M_{jk}=(p_{i'}-p_{j'})_y$ if $j=2$. As a result, $A_{jk}\geq \frac12|M_{jk}|>0$.

Case 2. $i\in\{1,2\}$, $j,k>2$. $(i,j),(i,k)\in E(A)$ implies that there exist $i'\in\mathcal{A}$, $j'=j-2+n_a\in\mathcal{S}$ and $k'=k-2+n_a\in\mathcal{S}$ such that $(i,j),(i,k)\in\mathcal{E}$. It follows that 
\[
\begin{split}
M&=(f_{i'}-f_{j'})(f_{i'}-f_{k'})^{\top}=\begin{pmatrix}
p_{i'}\\-e_{j-2}
\end{pmatrix}(p_{i'}^{\top},-e_{k-2}^{\top})\\
&=\begin{pmatrix}
p_{i'}p_{i'}^{\top} & -p_{i'}e_{k-2}^{\top}\\
-e_{j-2}p_{i'}^{\top} & e_{j-2}e_{k-2}^{\top}
\end{pmatrix}.
\end{split}
\]
Since $M_{jk}=1$, we have $A_{jk}\geq \frac12|M_{jk}|=\frac12$.

Case 3. $i,j,k>2$. Let $i'=i-2+n_a$, $j'=j-2+n_a$ and $k'=k-2+n_a$, we have 
\[
\begin{split}
M&=(f_{i'}-f_{j'})(f_{i'}-f_{k'})^{\top}\\
&=\begin{pmatrix}
\mathbf{0}_{2\times1}\\e_{i-2}-e_{j-2}
\end{pmatrix}(\mathbf{0}_{1\times2},e_{i-2}^{\top}-e_{k-2}^{\top})\\
&=\begin{pmatrix}
\mathbf{0}_{2\times2} & -p_{i'}e_{k-2}^{\top}\\
-e_{j-2}p_{i'}^{\top} & e_{j-2}e_{k-2}^{\top}
\end{pmatrix}.
\end{split}
\] 
Similar to Case 2, $A_{jk}\geq \frac12|M_{jk}|=\frac12$.	
\QEDA

Proof of Theorem \ref{th decomposeD}: From the proof of Theorem \ref{th D rank1}, one can realize that in the absence of the rank constraint on $D$, if the $3\times3$ submatrix corresponding to a triangle (e.g., composed of $i$, $j$ and $k$) is positive semi-definite, then constraints on angles in this triangle are exact (being equalities rather than inequalities). Moreover, if the $3\times3$ submatrix corresponding to a pair of angles sharing a common edge is positive semi-definite, then the corresponding three angle constraints are exact. For example, suppose $(i,j),(i,k),(i,h)\in\hat{\mathcal{E}}$, and $l_{ij}<l_{ik}<l_{ih}$, if the third order principal submatrix of $D$ corresponding to $l_{ij}$, $l_{ik}$, $l_{ih}$ is positive semi-definite, then $\frac{(x_i-x_j)^{\top}}{||x_i-x_j||}\frac{(x_i-x_k)}{||x_i-x_k||}=a_{ijk}$, $\frac{(x_i-x_k)^{\top}}{||x_i-x_k||}\frac{(x_i-x_h)}{||x_i-x_h||}=a_{ikh}$, $\frac{(x_i-x_j)^{\top}}{||x_i-x_j||}\frac{(x_i-x_h)}{||x_i-x_h||}=a_{ijh}$. This implies that if angles within each triangle are exactly constrained, then angles between edges in different triangles can also be exactly constrained. Note that for any $(i,j),(i,k),(i,h)\in\hat{\mathcal{E}}$, $l_{ij}$, $l_{ik}$ and $l_{ih}$ must be adjacent to each other in graph $\mathcal{G}(B)$. Moreover, the three edges of each triangle are also adjacent to each other in graph $\mathcal{G}(B)$. Hence, we only require the third order principal submatrix of $D$ corresponding to each 3-point clique in $\mathcal{G}(B)$ to be positive semi-definite, which must hold if $D_i=Q_{\mathcal{C}_i(B)}DQ_{\mathcal{C}_i(B)}^{\top}\succeq0$ for all maximal cliques $\mathcal{C}_i(B)$ of graph $\mathcal{G}(B)$.
\QEDA

\section{Acknowledgement}

The authors would like to thank the anonymous reviewers for their comprehensive comments and constructive suggestions on how to improve this paper. The authors also thank Prof. Brian D. O. Anderson and Prof. Shiyu Zhao for insightful conversations.

\vspace{-1cm}
\begin{IEEEbiography}[{\includegraphics[width=1in,height=1.25in,clip,keepaspectratio]{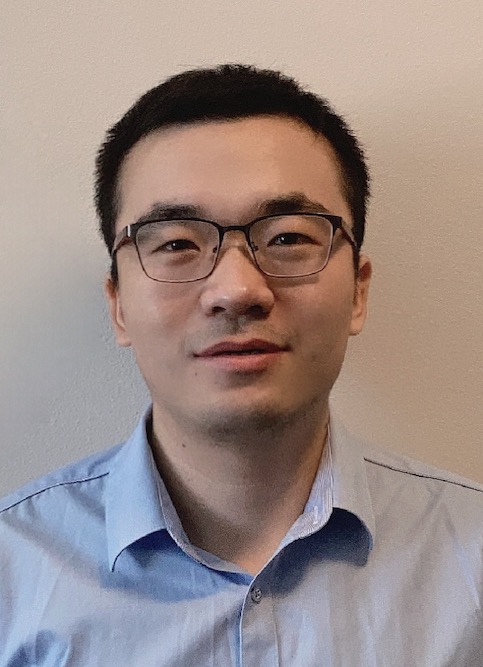}}]{Gangshan Jing}
	received the Ph.D. degree in Control Theory and Control Engineering from Xidian University, Xi'an, China, in 2018. From Dec. 2016- May. 2017, and Nov. 2017- Jan. 2018, he was a research assistant at  Department of Applied Mathematics, Hong Kong Polytechnic University, Hong Kong. From Oct. 2018- Sept. 2019, he was a postdoctoral researcher at Department of Mechanical and Aerospace Engineering, The Ohio State University, USA. Since Sept. 2019,  he has been a postdoctoral researcher at Department of Electrical and Computer Engineering, North Carolina State University, USA. His current research interests include control, optimization, and machine learning for network systems.
\end{IEEEbiography}

\begin{IEEEbiography}[{\includegraphics[width=1in,height=1.25in,clip,keepaspectratio]{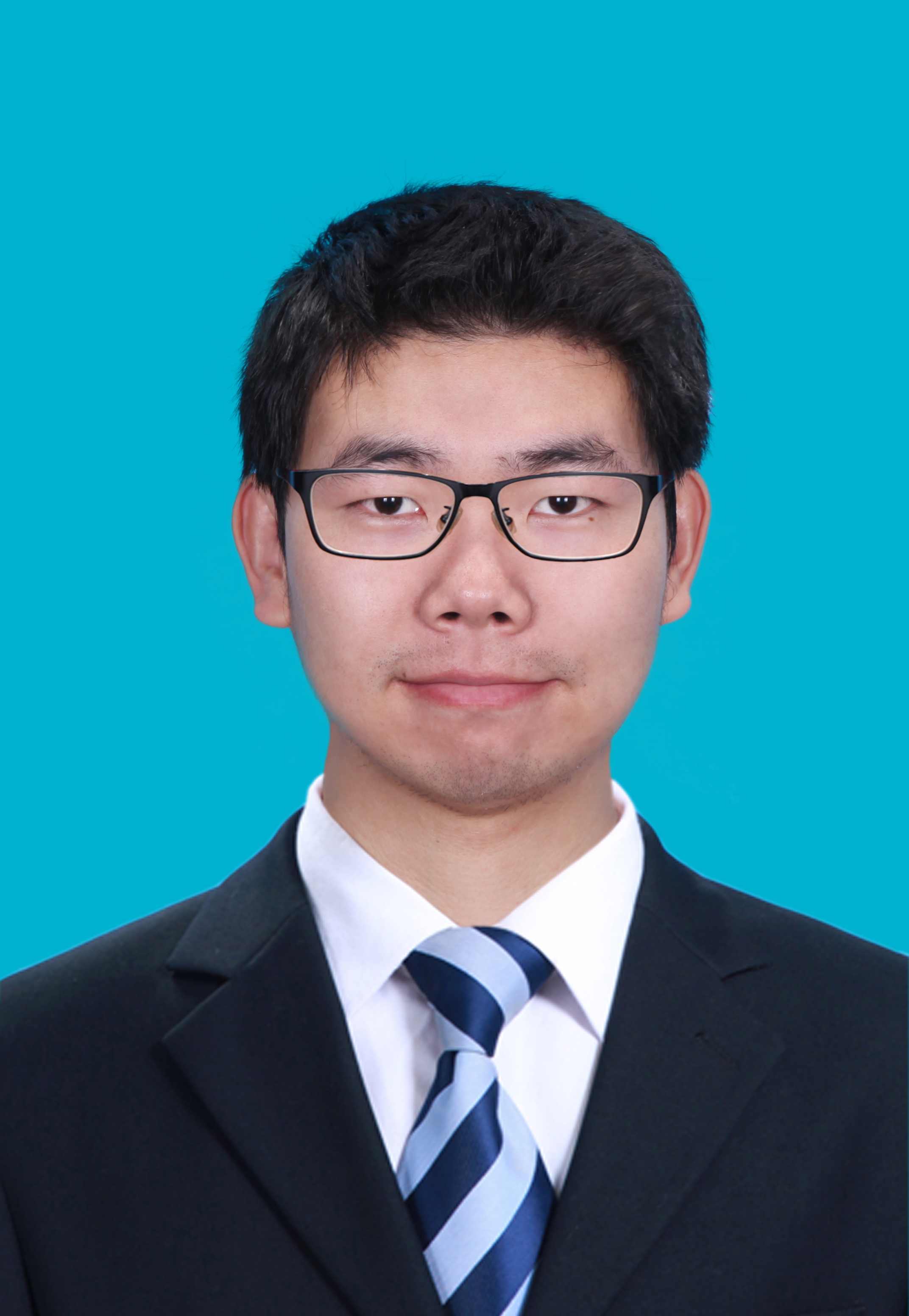}}]{Changhuang (Charlie) Wan}
	received his bachelor and master degrees in Spacecraft Design and Engineering from Beihang University, Beijing, China, in 2013 and 2016, respectively. He is currently working towards the Ph.D. degree in the Mechanical and Aerospace Engineering Department at The Ohio State University, Columbus, OH. His research interests include numerical optimization and autonomous systems.
\end{IEEEbiography}

\begin{IEEEbiography}[{\includegraphics[width=1in,height=1.25in,clip,keepaspectratio]{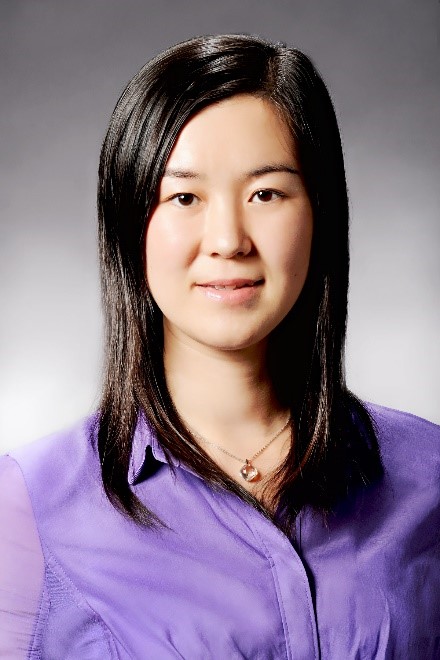}}]{Ran Dai}
	is an associate professor in School of Aeronautics and Astronautics at Purdue University. She received her B.S. degree in Automation Science from Beihang University and her M.S. and Ph.D. degrees in Aerospace Engineering from Auburn University. Dr. Dai's research focuses on control of autonomous systems, numerical optimization, and networked dynamical systems. She is an associate editor of IEEE transaction on Aerospace and Electronic Systems, and a recipient of the National Science Foundation Career Award and NASA Early Faculty Career Award.
\end{IEEEbiography}

\end{document}